\def\stoc{0}
\def\comments{0}

\ifnum\stoc=0
\documentclass[fullpage, p11]{article}
\else
\documentclass[screen,sigconf]{acmart}
\usepackage{microtype}
\fi

\ifnum\stoc=0
\usepackage[utf8]{inputenc}

\usepackage[colorlinks=true,linkcolor=green,citecolor=blue,allcolors=green,hypertexnames=false]{hyperref}
\hypersetup{
  colorlinks=true,
   linkcolor=blue,
   filecolor=blue,
   urlcolor=blue,
  allcolors=blue
}
\usepackage{fullpage}
\fi

\usepackage{amsfonts,amsmath}
\usepackage{graphicx,color,xcolor}
\ifnum\stoc=0
\usepackage{paralist}
\fi
\usepackage{enumitem}
\usepackage{algorithm}
\ifnum\stoc=0
\usepackage{algpseudocode}
\fi
\ifnum\stoc=0
\usepackage{sansmath}
\usepackage{array}
\usepackage{ragged2e}
\fi
\usepackage{cleveref}

\ifnum\stoc=0
\usepackage{thmtools}

\fi

\newcommand{\newcomment}[1]{\hfill  \textcolor{gray}
{/* \emph{\small #1} */} }

\newcommand{\eps}{\epsilon}

\newcommand{\davg}{d_{avg}}
\newcommand{\poly}{\mathrm{poly}}

\newcommand{\Ex}{\mathbb{E}}

\renewcommand{\th}{\textsuperscript{th}}

\newcommand{\F}{\textsf{F}}
\newcommand{\K}{\textsf{K}}
\newcommand{\C}{\textsf{C}}
\newcommand{\Cy}{\textsf{c}}
\newcommand{\Cyk}{{\textsf{C}}_{k}}
\newcommand{\cyk}{n_{\C}}

\newcommand{\cl}{n^{\leq \ogamma}_{\F}}
\newcommand{\cmh}{n^{> \ogamma}_{\F}}
\newcommand{\Bl}{B^{\leq \ogamma}}
\newcommand{\Bmh}{B^{> \ogamma}}
\newcommand{\Bm}{B^{(\ogamma,\sqrt{\om}]}}
\newcommand{\Bh}{B^{> \sqrt{\om}}}
\newcommand{\clk}{n^{\leq \ogamma}_k}
\newcommand{\cmhk}{n^{> \ogamma}_k}
\newcommand{\plk}{p^{\leq \ogamma}_k}
\newcommand{\pmhk}{p^{> \ogamma}_k}
\newcommand{\f}{\textsf{f}}
\newcommand{\Q}{G_{\Cy}}
\newcommand{\onC}{\overline{n}_{\C}}

\newcommand{\calA}{{\mathcal{A}}}
\newcommand{\calU}{{\mathcal{U}}}
\newcommand{\mA}{{\mathcal{A}}}

\newcommand{\calD}{{\mathcal{D}}}

\newcommand{\oeps}{\overline{\eps}}

\newcommand{\ogamma}{\overline{\gamma}}
\newcommand{\hgamma}{\widehat{\gamma}}
\newcommand{\od}{\overline{d}_{avg}}

\newcommand{\om}{\overline{m}}
\newcommand{\hm}{\widehat{m}}
\newcommand{\kappaF}{\kappa_{\F}}
\newcommand{\vecV}{\overrightarrow{V}}
\newcommand{\hnF}{\widehat{n}_{\F}}
\newcommand{\countF}{n_{\F}}
\newcommand{\nT}{n_{\textsf{T}}}

\newcommand{\hF}{h_{\F}}
\renewcommand{\S}{\textsf{S}}
\newcommand{\countS}{n_{\S}}
\newcommand{\countK}{n_{\scriptscriptstyle \K}}

\newcommand{\onF}{\overline{n}_{\F}}

\newcommand{\EX}{\mathbb{E}}
\newcommand{\ham}{\mathcal{H}}
\newcommand{\countC}{n_{\C}}

\ifnum\stoc=0
\newenvironment{proof}{\smallskip\noindent{\bf Proof:}}%
       {\hspace*{\fill}$\Box$\par}
\fi

\ifnum\stoc=0
\newenvironment{proofof}[1]{\smallskip\noindent{\bf Proof of #1:}}
        {\hspace*{\fill}$\Box$\par}
\else
\newenvironment{proofof}[1]{\smallskip\noindent{\sc Proof of #1.}}
        {\hspace*{\fill}$\Box$\par}
\fi

\ifnum\stoc=0
\newtheorem{theorem}{Theorem}
\newtheorem{thm}{Theorem}

\newtheorem{lemma}{Lemma}

\newtheorem{definition}{Definition}
\fi

\ifnum\comments=1
\newcommand{\dana}[1]{ {\color{magenta} {\textbf [D: {#1}]}}}
\newcommand{\talya}[1]{ {\textbf{\color{red} [T: #1]}}}
\newcommand{\tnew}[1]{{\color{red} #1}}
\newcommand{\tfootnote}[1]{\footnote{\textcolor{red}{#1}}}
\newcommand{\dnew}[1]{{\color{magenta} #1}}

\newcommand{\old}[1]{}
\newcommand{\told}[1]{\textcolor{gray}{#1} }

\newcommand{\ronitt}[1]{\textbf{\color{purple} [Ronitt: #1]}}
\newcommand{\reut}[1]{ {\textbf{\color{cyan} [R: #1]}}}

\newcommand{\SodaReb}[1]{{\color{blue} #1}}

\newcommand{\ForFuture}[1]{{\color{red} \textbf{TODO:} #1}}

\else

\newcommand{\dana}[1]{}
\newcommand{\talya}[1]{}
\newcommand{\tnew}[1]{{#1}}
\newcommand{\tfootnote}[1]{}
\newcommand{\dnew}[1]{{#1}}

\newcommand{\old}[1]{}
\newcommand{\told}[1]{}

\newcommand{\ronitt}[1]{}
\newcommand{\reut}[1]{}
\newcommand{\SodaReb}[1]{}
\newcommand{\ForFuture}[1]{}
\fi
\makeatletter
\def\moverlay{\mathpalette\mov@rlay}
\def\mov@rlay#1#2{\leavevmode\vtop{%
   \baselineskip\z@skip \lineskiplimit-\maxdimen
   \ialign{\hfil$\m@th#1##$\hfil\cr#2\crcr}}}
\newcommand{\charfusion}[3][\mathord]{
    #1{\ifx#1\mathop\vphantom{#2}\fi
        \mathpalette\mov@rlay{#2\cr#3}
      }
    \ifx#1\mathop\expandafter\displaylimits\fi}
\makeatother

\ifnum\stoc=0
\newcommand{\preprocess}{\hyperref[alg:preprocess]{\color{black}{\sf Construct-Data-Structure}}}
\else
\newcommand{\preprocess}{{\color{black}{\sf Construct-Data-Structure}}}
\fi

\ifnum\stoc=0 \newcommand{\sampHighVertex}{\hyperref[alg:sampHigh]{\color{black}{\sf Sample-Medium-High-Vertex}}}
\else
\newcommand{\sampHighVertex}{{\color{black}{\sf Sample-Medium-High-Vertex}}}
\fi

\newcommand{\sampLowCopies}{\hyperref[alg:sampLowCopies]{\color{black}{\sf Sample-Low-Copy}}}

\newcommand{\sampHMCopies}{\hyperref[alg:sampHMCopies]
{\color{black}{\sf Sample-Mixed-Copy}}}

\newcommand{\sampHighCliques}{\hyperref[alg:sampHighCliques]
{\color{black}{\sf Sample-High-Clique}}}

\newcommand{\sampMediumCliques}{\hyperref[alg:sampMediumClique]
{\color{black}{\sf Sample-Medium-Clique}}}

\newcommand{\sampLowCliques}{\hyperref[alg:sampLowClique]
{\color{black}{\sf Sample-Low-Clique}}}

\newcommand{\ApproxWithEstimate}
{\hyperref[alg:approx-with-estimate]{\color{black}{\sf Approximate-Counting-with-Advice}}}

\newcommand{\SampleWithEstimate}
{\hyperref[alg:sample-copy]{\color{black}{\sf Sample-Copy}}}

\newcommand{\SampleCliqueWithEstimate}
{\hyperref[alg:sample-clique]{\color{black}{\sf Sample-Clique}}}

\newcommand{\Approx}
{\hyperref[alg:approx-search]{\color{black}{\sf Approximate-Counting}}}

\newcommand{\sets}{\lceil \frac{n}{\ogamma}\cdot \frac{3\ln(4nt/\delta)}{\oeps^2}\rceil}

 \ifnum\stoc=1
\newcommand{\alg}[2]{
        \begin{flushleft}
        \fbox{\parbox{0.95\linewidth}{
	#2
        }}
        \end{flushleft}
        \label{#1}
}	
\else
\newcommand{\alg}[2]{
 \begin{figure*}[htb]
	\centering
        \fbox{\parbox{0.95\linewidth}{
	#2
        }}
        \label{#1}
\end{figure*}
}
\fi

\newcommand{\search}{{\sf Search}}
\newcommand{\algA}{\calA}
\newcommand{\ApproxCliquesEstimate}{\textsf{Approx-Count-Cliques-With-Estimate}}
\newcommand{\ok}{\overline{n}_{\textsf{K}}}
\newcommand{\hk}{\hat{n}_{\textsf{K}}}
\newcommand{\invokeA}{\mathcal{A}}
\newcommand{\ov}{\overline{v}}
\newcommand{\vvv}{v}
\newcommand{\hv}{\hat{v}}
\newcommand{\ert}[1]{E_{rt}\left(#1\right)}

\newcommand{\nS}{n_{\textsf{S}}}
 \newcommand{\onS}{\overline{n}_{\textsf{S}}}
 \newcommand{\odmax}{\overline{d}_{max}}

\ifnum\stoc=1

\setcopyright{cc}
\setcctype{by}
\acmDOI{10.1145/3717823.3718160}
\acmYear{2025}
\copyrightyear{2025}
\acmISBN{979-8-4007-1510-5/25/06}
\acmConference[STOC '25]{Proceedings of the 57th Annual ACM Symposium on Theory of Computing}{June 23--27, 2025}{Prague, Czechia}
\acmBooktitle{Proceedings of the 57th Annual ACM Symposium on Theory of Computing (STOC '25), June 23--27, 2025, Prague, Czechia}
\received{2024-11-04}
\received[accepted]{2025-02-01}

\fi

\begin{document}

\ifnum\stoc=1
\title{Approximately Counting and Sampling Hamiltonian Motifs in Sublinear Time}\titlenote{
The second author was supported by the ISF grant 1867/20.
The last author was supported by the NSF TRIPODS program (award DMS-2022448) and CCF-2310818.
Part of this work was conducted while the first and last authors were visiting the Simons
Institute for the Theory of Computing as part of the Sublinear Algorithms program and while last author was on Sabbatical at Tel-Aviv University.
}
\fi

\ifnum\stoc=1
\author{Talya Eden}
\orcid{0000-0001-8470-9508}
\affiliation{%
  \institution{Bar-Ilan University}
  \city{Ramat Gan}
  \country{Israel}
}
\email{talyaa01@gmail.com}

\author{Reut Levi}
\orcid{0000-0003-3167-1766}
\affiliation{%
  \institution{Reichman University}
  \city{Herzliya}
  \country{Israel}
}
\email{reut.levi1@runi.ac.il}

\author{Dana Ron}
\orcid{0000-0001-6576-7200}
\affiliation{%
  \institution{Tel Aviv University}
  \city{Tel Aviv}
  \country{Israel}
}
\email{danaron@tauex.tau.ac.il}

\author{Ronitt Rubinfeld}
\orcid{0000-0002-4353-7639}
\affiliation{%
  \institution{Massachusetts Institute of Technology}
  \city{Cambridge}
  \country{USA}
}
\email{ronitt@csail.mit.edu}

\fi

\ifnum\stoc=0
\title{Approximately Counting and Sampling Hamiltonian Motifs in Sublinear Time\thanks{Part of this work was conducted while the first and last authors were visiting the Simons
Institute for the Theory of Computing as part of the Sublinear Algorithms program.}}
	\author{
	Talya Eden\thanks{Department of Computer Science, Bar Ilan University, Israel. Email: {\tt talyaa01@gmail.com}}
		\and
		Reut Levi\thanks{Efi Arazi School of Computer Science,  Reichman University, Israel. Email: {\tt reut.levi1@runi.ac.il}. Supported by the ISF grant 1867/20}
			\and
		Dana Ron\thanks{School of Electrical and Computer Engineering, Tel Aviv University, Israel. Email: {\tt danaron@tau.ac.il}.}
        \and
        		Ronitt Rubinfeld\thanks{Computer Science and Artificial Intelligence Laboratory,  MIT, USA. Email: {\tt ronitt@mit.edu}. Supported by the NSF TRIPODS program (award DMS-2022448) and CCF-2310818. Part of this work was conducted while the author was on Sabbatical at Tel Aviv University}
}
\begin{titlepage}
\maketitle
\fi

\begin{abstract}

Counting small subgraphs, referred to as \emph{motifs}, in large graphs is a fundamental task in graph analysis, extensively studied across various contexts and computational models. 
In the sublinear-time regime, 
the relaxed  problem of approximate counting has been explored within two prominent query frameworks: 
the \emph{standard model}, which permits degree, neighbor, and pair queries, and the strictly more 
powerful \emph{augmented model}, which additionally allows for uniform edge sampling. 
Currently, in the standard model, 
(optimal)  results  have been  established  
only 
for approximately counting edges, stars, and cliques, all of which have a radius of one. This contrasts sharply
with the state of affairs in the  augmented model, where 
algorithmic results (some of which are optimal) are known for any input motif, leading to a 
disparity which we term the ``scope gap" between the two models.

In this work, we  make significant progress in bridging this gap. 
Our approach 
draws inspiration from recent advancements in the augmented model 
and utilizes a framework centered on counting by uniform sampling, thus 
allowing us to establish new results in the standard model and simplify on previous results.

In particular, our first, and main, contribution is  a new algorithm in the standard model for approximately counting \emph{any Hamiltonian motif} in sublinear time, 
where the complexity of the algorithm is the sum of two terms. One term equals the complexity of the  known algorithms by Assadi,  Kapralov, and  Khanna (ITCS 2019) and Fichtenberger and Peng (ICALP 2020)  in the (strictly stronger) augmented model 
and the other is 
an  additional, necessary,
additive overhead. 

Our second contribution is a variant of our algorithm that enables nearly uniform sampling of these motifs, a capability previously limited in the standard model to edges and cliques. 
Our third contribution is to introduce even simpler algorithms for stars and cliques by exploiting their radius-one property. As a result, we simplify all previously known algorithms in the standard 
model 
 for stars (Gonen, Ron, Shavitt (SODA 2010)), triangles (Eden, Levi, Ron Seshadhri (FOCS 2015)) and cliques (Eden, Ron, Seshadri (STOC 2018)).
\end{abstract}

\ifnum\stoc=1



\maketitle

\fi

\ifnum\stoc=0
\thispagestyle{empty}
\end{titlepage}
\pagenumbering{roman}
\tableofcontents\newpage
\pagenumbering{arabic}
\newpage
\fi

\section{Introduction}
Given a graph $G= (V,E)$ over $n$ vertices and $m$ edges, and a 
(small)
graph $\F$ over $k$ vertices, which we refer to as a \emph{motif},\footnote{\label{foot:motif} The term ``motif'' usually refers to subgraphs that are more recurrent than expected in a random graph, but here we use it for any small 
subgraph of interest.} 
 we are interested in the computational tasks  of (approximately) counting the number of \emph{copies} of $\F$ in $G$ and (approximately) uniformly sampling such copies. A copy of $\F$ in $G$ is a subgraph of $G$ that is isomorphic to $\F$, and we denote the number of copies by $\countF(G)$ or simply $\countF$ when $G$ is clear from the context.


Counting  motifs is crucial in numerous network analysis applications  and is used to study networks across various fields, encompassing biology, chemistry, social studies, among others (see e.g.,~\cite{Alon07, hormozdiari2007not, prvzulj2004modeling, burt2004structural, milo2002network, benson2016higher, swiegers2000new}).
For exact counting of general motifs, devising an algorithm that runs in time $f(k) \cdot n^{o(k)}$, for any function $f$, is unlikely.\footnote{\label{foot:shorp-W} This is due to the fact that $k$-clique counting is a \#W[1]-hard problem~\cite{SharpW1-hard}.
} Furthermore, in many real-world settings, 
even a linear dependence on $n$ is prohibitive. A naive sampling algorithm can be applied to get an approximation of $\countF$ for any $\F$ in time $O(n^k/\countF)$, implying sublinear complexity for sufficiently large values of $\countF$.
A natural question is whether it is possible to significantly improve upon the naive sampling algorithm, achieving lower sublinear complexity and extending the range of $\countF$ for which sublinear upper bounds are applicable. 
Before addressing this question, we discuss two models for accessing the graph in the sublinear regime, which have been the primary focus for studying problems in this area.

We refer to the first model as the 
\emph{standard model}~\cite{PR,KKR04}, where the algorithm can perform the following types of queries:
\ifnum\stoc=0
\begin{inparaenum}[(1)] 
\item Degree query: for any vertex $v\in V$, returns $v$'s degree $d(v)$;
\item Neighbor query: for any vertex $v\in V$ and index $i$, 
returns $v$'s $i\th$ neighbor (if $i> d(v)$, then a special symbol is returned);\footnote{ \label{deg} Indeed a degree query to vertex $v$ can be emulated by $O(\log(d(v)))$ neighbor queries.}
\item Pair query: for any pair of vertices $u,v$, returns whether $\{u,v\}\in E$.
\end{inparaenum}
\else
(1) Degree query: for any vertex $v\in V$, returns $v$'s degree $d(v)$;
(2) Neighbor query: for any vertex $v\in V$ and index $i$, 
returns $v$'s $i\th$ neighbor (if $i> d(v)$, then a special symbol is returned);\footnote{ \label{deg} Indeed a degree query to  $v$ can be emulated by $O(\log(d(v)))$ neighbor queries.}
(3) Pair query: for any pair of vertices $u,v$, returns whether $\{u,v\}\in E$.
\fi
The second model, which we refer to as the \emph{augmented model}~\cite{Aliak}, allows for the above queries as well as 
(4) Uniform edge query: returns  a uniformly chosen edge in the graph.  
It was shown in~\cite{Aliak}, that the augmented model is strictly stronger.

We first describe the state of the art for approximate counting in these models
and later  discuss sampling. 
In what follows, the complexity bounds we state are in expectation. 
Current results in the standard model are known for  approximately counting edges~\cite{GR08, Feige-Avg}, stars~\cite{GRS11}, triangles~\cite{ELRS} and $k$-cliques~\cite{cliquesSicomp}.
 Specifically, for $k$-cliques, including edges ($k=2$) and triangles ($k=3$), there is an upper bound of $O^*\Big(\frac{n}{\countF^{1/k}} +\frac{m^{k/2}}{\countF}\Big)$,\footnote{\label{foot:star} We use $O^*(\cdot)$  to suppress  factors of $O(k!)$, and $\poly(1/\eps,\log n)$, where $\eps$  is the approximation parameter.} 
  which is essentially tight as long as this expression is upper bounded by $m$.
On the other hand, in the  augmented model, there are known results for any 
input motif~\cite{Aliak, AKK, fichtenberger2020sampling, BER}.
The  complexity of~\cite{AKK, fichtenberger2020sampling}'s algorithm is $O^*\left(\frac{m^{\rho(\F)}}{\countF}\right)$ where $\rho(\F)$ is the fractional edge-cover number\footnote{\label{foot-frac-cover} A fractional edge cover of a graph $Q = (V(Q),E(Q))$ is
a function $w:E(Q) \to [0,1]$, such that $\sum_{e\in E_v} w(v)\geq 1$ for every vertex $v\in V(Q)$, where $E_v$ is the set of edges incident to $v$. The fractional edge-cover number of $Q$, denoted $\rho(Q)$, is the minimum, over all edge covers $w$, of $\sum_{e\in E(Q)}w(e)$.} of $\F$. In particular, for $k$-cliques $\rho(\F) = k/2$, so that
the complexity is $O^*\left(\frac{m^{k/2}}{\countF}\right)$, and there is also a matching lower bound  (that holds 
as long as the expression is upper bounded by $m$).

Thus, there are two gaps between the models. The first, which we refer to as the \emph{complexity gap}, is unavoidable, and stems from the stronger 
power that uniform edge queries allow. The second, which we refer to as the \emph{scope  gap},  is between the set of results known for each model. 
In the standard model there are  known 
results only for  specific motifs.  These motifs 
all have radius one, a property
that is used in an essential way in the aforementioned algorithms.
In contrast, in the augmented model there is an 
algorithm for any given motif (albeit not-necessarily optimal).

One possibility to bridge the scope gap, is to simulate each 
uniform edge query in the augmented model by a procedure, in the standard model, that returns a uniformly distributed edge. The question of sampling uniform edges in the standard model 
was investigated in a series of works~\cite{ER17,TT_edges,EMR_multiple_edges,ENT_edge} 
where the state of the art is an optimal algorithm for sampling a single edge exactly uniformly~\cite{ENT_edge} with  complexity $\Theta\big(\frac{n}{\sqrt m}\big)$, and an  algorithm
for sampling $r$ edges from an $\eps$-pointwise\footnote{\label{foot:eps-pointwise} A distribution $P$ is $\eps$-pointwise close to uniform, if  $\max_{x\in \calU} \lvert p(x) -\frac{1}{|\calU|}\rvert<\frac{\eps}{|\calU|}$. 
} close to uniform distribution with complexity $\Theta^*\big(\sqrt{r}\cdot \frac{n}{\sqrt m}\big)$~\cite{EMR_multiple_edges}. By~\cite{TT_edges}, this is essentially optimal.  Hence, any  algorithm in the  augmented model that performs $O(q)$ queries can be simulated by an algorithm in the standard model, resulting in  complexity $O^*\big(q+\sqrt{q}\cdot \frac{n}{\sqrt m}\big)$.
However, in many cases, this 
overhead becomes prohibitive,
ultimately defeating the purpose of achieving sub-linearity.\footnote{\label{foot:unif-edge-sim} 
We note that by performing $n$  queries 
it is possible to implement uniform edge queries. This is done by querying the degrees of all vertices and constructing a 
data structure that allows  to sample vertices with probability proportional to their degree.
} 
Thus, in this paper we ask the following question:

\emph{Can we bridge the scope gap between the  motif counting  results in the augmented model and the standard model with 
only the necessary overhead?}

As stated in the next theorem, we answer the above question in the affirmative for every motif $\F$ that is Hamiltonian.
For such motifs the complexity  in the augmented model is $O^*\big(\frac{m^{k/2}}{\countF}\big)$.
Our algorithm incurs an additive overhead of $O^*\big(\frac{n}{\countF^{1/k}}\big)$, which by a result of~\cite[Theorem 4.1]{ER18_lbs} is indeed necessary. 
In particular, these upper bounds (in both models)  are optimal for $k$-cliques~\cite{ERR19} and odd-length $k$-cycles~\cite{AKK}.

\begin{theorem}
\label{thm:main}
Let $G$ be a graph over $n$ vertices and $m$ edges.
    There exists an algorithm in the standard query model that, given query access to  $G$ and parameters $n$, $\eps\in(0,1)$ and 
    a Hamiltonian motif $\F$ over $k$ vertices, returns a value $\hnF$ such that  $\hnF\in(1\pm\eps)\countF$ with probability at least $2/3$. 
%
The expected query complexity and running time of the algorithm are  
$O^*\Big(\frac{n}{\countF^{1/k}}+ 
    \frac{m^{k/2}}{\countF}
    \Big).$\footnote{\label{foot:thm1-com} A few notes are in place.  
 The first is that, as is standard, 
 the success probability of the algorithm can be increased to $1-\delta$ at a multiplicative cost of $\log(1/\delta)$.
The second is that the complexity bound in Theorem~\ref{thm:main} also holds with high probability.
The third is that it can be ensured that the query complexity of the algorithm never exceeds $O(n+m)$. 
}
\end{theorem}

\paragraph{Pointwise-close-to-uniform sampling.}\label{subsec:pointwise}
Sampling motifs enables further exploration of their properties and interactions beyond just their (approximate) count.
The techniques used in our counting algorithm can also be applied to sample Hamiltonian motifs from a distribution that is  \emph{$\eps$-pointwise close} to uniform. 
That is, a distribution where every motif  has sampling probability in $\frac{(1\pm\eps)}{\countF}.$ 
This is a stronger guarantee than the $\ell_1$-closeness guarantee, 
since $\ell_1$-closeness allows  to ignore an $\eps$-fraction of the motifs, thus resulting in a potentially unrepresentative sample. 

In the augmented model, 
it is possible to obtain exact uniform sampling of any motif
with the same respective complexities as for counting~\cite{fichtenberger2020sampling, BER}. 
In the standard model, the 
known results are for 
sampling edges exactly uniformly~\cite{ENT_edge}, and 
for sampling $k$-cliques pointwise-close to 
uniformly~\cite{ERR_sampling_cliques}.  

As with  approximate counting, we bridge the scope gap for every Hamiltonian motif, with only the necessary overhead.

    \begin{theorem}\label{thm:sampling}
 Let $G$ be a graph over $n$ vertices and $m$ edges.
    There exists an algorithm in the standard query model that, given query access to  $G$ and parameters $n$, $\eps,\delta\in(0,1)$ and 
    a Hamiltonian motif $\F$ over $k$ vertices,
    satisfies the following.
    It performs a preprocessing step such that with probability at least $1-\delta$, following this step it is possible to independently sample copies of $\F$ in $G$ according to a distribution that is $\eps$-pointwise-close to uniform.
The expected query complexity and running time of the preprocessing step, as well as the number of queries and running time sufficient to obtain each sample,
are
$O^*\Big(\frac{n}{\countF^{1/k}}+ 
    \frac{m^{k/2}}{\countF}
   \Big)\cdot \log(1/\delta)\;.$\footnote{The second and third comment in Footnote~\ref{foot:thm1-com} hold here as well.}
\end{theorem}

We note that our algorithms (like most prior ones)  can be easily adapted to work (with the same complexity) for approximately counting and sampling directed motifs in directed graphs.
We elaborate more on this extension in 
\ifnum\stoc=0
Appendix~\ref{app:ext}.
\else
the full version~\cite{ELRR25-arxiv}.
\fi


\paragraph{Simple algorithms for counting and sampling cliques.} 
Triangles, and more generally $k$-cliques, have generated much interest in various fields of study.
For example, they can be used to predict protein complexes and regulatory sites, predict missing links,  detect dense network regions and potential communities, detect spammers and more~\cite{milo2002network,sole2006network, palla2005uncovering,prat2016put,becchetti2008efficient}. 

Applying our algorithms to the special case when $\F$ is a $k$-clique gives  simpler algorithms with the same 
(optimal) complexity as in previous  sublinear algorithms for approximate counting~\cite{ELRS,cliquesSicomp} and sampling~\cite{ERR_sampling_cliques}. 
Moreover, using the approach presented in this work, we describe 
\emph{even simpler} 
algorithms for $k$-cliques, by exploiting the fact that a clique has diameter one.

\paragraph{Simple algorithms for counting and sampling  stars.} 
Approximately counting $k$-stars (a central vertex connected to $k-1$ leaves)  is equivalent up to normalization to estimating the $(k-1)\th$-moment of the degree distribution of a graph. This task has important implications  to
 characterizing and  modeling networks in varied applications~\cite{Faloutsos01,pennock2002winners,sala2010measurement,bickel2011method}.
Similarly, sampling a nearly uniform star can be easily modified into  sampling vertices 
 according to the $(k-1)\th$ moment of the degree distribution, so that each vertex is sampled with probability $\frac{d(v)^{k-1}}{\nS}$, where $\nS$ denotes the number of $k$-stars. 
 This is the extensively studied $\ell_p$-sampling problem for $p=k-1$ (see a recent survey and references therein~\cite{lp-survey}). 

Stars are clearly not Hamiltonian, and hence we cannot apply our algorithms for Hamiltonian motifs to approximately count and sample them. However, similarly to cliques, we design tailored simple algorithms for stars. We
get the same (optimal) complexity  of
$
O^*\left(\frac{n}{\countS^{1/k}}+\min\left\{\frac{m\cdot n^{k-2}}{\countS},\frac{m}{\countS^{1/(k-1)}}\right\}\right)$ for
approximately counting $k$-stars as in previous works~\cite{GRS11,ERS19}.
 We note that 
 the simplification  is more in terms of the  analysis than that of the algorithm. Furthermore, there were no known previous results for almost uniformly sampling of stars. Here too, our algorithm takes inspiration from results in the augmented model~\cite{BER}.

\subsection{A framework for sampling and counting based on attempted samplers} \label{sec:framework}

Our algorithms depart from the approach taken by most previous algorithms in the standard query model, which are often quite complex. Despite operating within the standard model, our algorithms exhibit a structural design more typical of those in the augmented model, which tend to be simpler and ``cleaner''.

We take the approach of approximate counting and nearly uniform sampling based on what we refer to as  \emph{attempted-samplers}. An attempted sampler either outputs a copy of the motif $\F$ or outputs `fail'.
The crucial point is that 
in every invocation of an attempted sampler, each copy is sampled with almost equal probability $p_{s}$, and this probability is exactly known to the algorithm.
 The overall success probability of a single sampling attempt is thus $\countF\cdot p_s$. 

Given such an attempted sampler and assuming we have a constant-factor estimate $\onF$ of $\countF$ (this assumption will be discussed momentarily),
by performing $O\Big(\frac{1}{p_s}\cdot \frac{1}{\eps^2 \cdot \onF}\Big)$ sampling attempts, we can get a $(1\pm\eps)$-approximation of $\countF$.  Since the number of repeated invocations grows linearly with $1/p_s$, 
the main challenge is in designing an attempted sampler that 
succeeds with sufficiently high probability, while still ensuring that every copy is (almost) equally likely to be returned.
%
 The assumption on 
 having a constant-factor estimate $\onF$ of $\countF$ can  
be removed by running a search algorithm, as shown e.g., in~\cite{cliquesSicomp}. This search does not increase the query complexity and running time of our algorithm by more than polylogarithmic factors and results in an approximate counting algorithm (that does not depend on any assumptions).

To achieve nearly uniform sampling—where the algorithm consistently returns a copy and the distribution over copies is pointwise close to uniform—we proceed as follows. 
We first run an approximate counting algorithm to obtain a good estimate of 
$\countF$, and  then 
repeatedly invoke the attempted sampler until a copy is successfully returned.
%
A subtle but crucial aspect here is that if the attempted sampler is run with inaccurate estimates, it may produce copies with a success probability $p_s$
  that is too low. In such cases, running the attempted sampler repeatedly until a copy is returned could lead to a prohibitively high expected runtime. First applying the approximate counting algorithm to obtain a good estimate and only then repeatedly invoking the attempted sampler until a copy is returned, ensures that the attempted sampler achieves the desired success probability $p_s$, hence keeping the expected runtime within the desired bounds.
  
To conclude, 
given an attempted sampler, the tasks of approximate counting and nearly uniformly sampling readily  follow, as will be discussed in more detail in Sections~\ref{subsubsec:intro-approx-count} and~\ref{subsec:intro-approximate-sampling}.  Thus,
the main challenge in our work is in the design of the attempted sampling procedures. 
We refer to the final sampling algorithm (that is promised to return a copy) as the \emph{final sampler}. 

\subsection{Attempted sampling procedures for Hamiltonian motifs}

For a given Hamiltonian motif $\F$,
we attempt to  sample copies of $\F$ in $G$ by first finding  Hamiltonian cycles of such copies.
For the sake of the exposition, we first focus on the case that the motif $\F$ is the graph $\Cyk$, that is, a cycle over $k$ vertices, and hence coincides with its single Hamiltonian cycle. 
In Section~\ref{subsubsec:Ham} we explain what needs to be modified for the more general case when $\F$ contains at least one such cycle  and possibly additional edges. 
We shall use the shorthand $\cyk$ for $n_{\Cyk}$.
\old{and the shorthand $\onC$ for $\overline{n}_{{\Cyk}}$.}

\sloppy Our starting point is the  
following attempted  sampler that works in the \emph{augmented model}~\cite{fichtenberger2020sampling}. For even-length cycles, it performs $k/2$ uniform edge queries, and obtains the edges $ \{v_1,v_2\},\{v_3,v_4\},\dots,\{v_{k-1},v_{k}\}$. It then performs all pair queries $\{v_i,v_{i+1}\}$ for all even $i$, as well as $\{v_k,v_1\}$ to check if  a cycle was obtained. 
The probability that any specific copy of
$\Cyk$  is sampled is exactly the same:
$\frac{2k}{ m^{k/2}}$, where the $2k$ factor is due to the different cyclic orderings in which the cycle's edges can be sampled.

It would be ideal if we could simulate the aforementioned process within the standard model. 
While, as previously noted, simulating each edge sample independently is impractical,
it is possible to efficiently sample edges that are incident to vertices of sufficiently high degree. 
For now, we consider the threshold $\sqrt m$, and refer to vertices with degree greater than $\sqrt m$ as high-degree vertices.
Sampling edges incident to high-degree vertices can be achieved  by applying a preprocessing step  from prior works~\cite{cliquesSicomp,EMR_multiple_edges}.
The preprocessing time  is $O(\frac{n}{\sqrt m})$  which is  upper bounded by the first term in our complexity.
Subsequently to the preprocessing step, it is possible to  sample, at unit cost, edges incident to high-degree vertices,  such that each edge is returned with probability roughly $\frac{1}{m}$. 
Hence, it is possible to simulate the 
augmented procedure for copies 
where every second edge is incident to a high-degree vertex.

This approach can be extended to copies containing at least one high-degree vertex by introducing  the concept of a \emph{path cover} -- a sequence of paths that collectively cover all vertices in a copy. Each path begins at a high-degree vertex and continues along vertices that are not of high degree. To sample such paths, we first sample the initial edge,  and then sample each subsequent vertex via neighbor queries,  each neighbor is sampled with probability $\frac{1}{\sqrt{m}}$, irrespective of the degree of the current vertex.
In cases where the path comprises of only a single high-degree vertex $v$, the preprocessing allows for sampling 
$v$ with probability $\frac{1}{\sqrt{m}}$. Therefore, we can sample each copy that contains at least one high-degree vertex by sampling all paths in its cover and verifying  that they collectively form a cycle using pair queries. Consequently, each such copy $\Cy$ is sampled with probability $\frac{\nu(\Cy)}{m^{k/2}}$, where $\nu(\Cy)$ denotes the number of possibilities to sample a path cover for this specific copy  (as further discussed in Section~\ref{subsec:samp-mixed-copy}). 
By the discussion in Section~\ref{sec:framework}, this probability determines the second term in our  bounds.

On the other hand, 
as we show in Section~\ref{subsubsec:low}, if all vertices in a copy have degree at most $\gamma = \countC^{1/k}$ (which is upper bounded by $\sqrt{m}$) then each such ``low copy'' can be sampled with probability $\frac{2k}{n\cdot \gamma^{k-1}}$. 
 By the discussion in~\ref{sec:framework}, this results in the first term in our complexity bound: $O^*\Big(\frac{n\cdot \gamma^{k-1}}{\countC}\Big)=O^*\Big(\frac{n\cdot \countC^{(k-1)/k}}{\countC}\Big)=O^*\Big(\frac{n}{\countC^{1/k}}\Big)$.
 Recall that this term is necessary for any motif over $k$ vertices,  by the lower bound of~\cite{ER18_lbs}.

The remaining challenge involves handling copies that have no high-degree vertices but include at least one medium-degree vertex -- defined as a vertex 
$v$ with degree $d(v)\in (\gamma, \sqrt{m}]$. To address these cases, we modify the path cover definition to allow paths to start with medium-degree vertices, provided they include at least two vertices (the reasoning behind this will be discussed shortly).
We also reduce the threshold used for the preprocessing step from $\sqrt{m}$ to 
$\gamma$. This results in a cost of  $O^*\big(\frac{n}{\gamma}\big)$, however, it does not asymptotically change our complexity, as it is of the same order as  the first term in our complexity bound $O^*\Big(\frac{n}{\countC^{1/k}}\Big)$.\footnote{Indeed, the choice of $\gamma$ is such that the two terms, $\frac{n}{\gamma}$ and $\frac{n\cdot \gamma^{k-1}}{\countC}$ are equated.}

With these modifications,  we can now sample edges incident to medium vertices, each with probability $\frac{1}{m}$, at a unit cost. If a  medium vertex is the first in a path,  we use its ``high enough" properties to sample an incident edge;   if it appears mid-path,  we use its ``low enough" properties, to uniformly sample a neighbor with probability $\frac{1}{\sqrt{m}}$. The only limitation of medium vertices  compared to  high vertices is that it is not possible to sample them uniformly with probability $\frac{1}{\sqrt{m}}$. 
However, this property of high vertices is only relevant when a path consists of a single high-degree vertex, and our path cover definition prevents paths from consisting of a single medium vertex. Consequently, we can accommodate all cases by applying one of the described methods, resulting in  two  samplers, summarized below.\footnote{We note that an alternative option is to design a separate, additional attempted sampler for copies that have only medium and low vertices. This attempted sampler would be a kind of ``hybrid'' between  the one for sampling copies that contain at least one high vertex and the one for sampling copies containing only low vertices. We opted for the more concise version of having only two attempted samplers.}

\subsection{More details on the  attempted samplers for low and mixed copies}
\label{subsubsec:low} \label{subsec:mixed}

Though we consider undirected graphs, it will be convenient to view each edge $\{u,v\}$ as two ordered edges, $(u,v)$ and $(v,u)$. We assume we have a constant-factor estimate $\om$ of the number of ordered edges, where such an estimate can be obtained using known algorithms~\cite{Feige-Avg,GR08}. We also assume that we have a constant-factor estimate $\onC$ of $\countC$ (where as noted previously, this assumption can be removed by an appropriate geometric search). 
Since we do not know $m$ and $\countC$ precisely, these two estimates determine our degree thresholds. For $\ogamma = \onC^{1/k}$, we say that a vertex $v$ has \emph{low} degree (or \emph{is low}) if $d(v) \leq \ogamma$. We say that $v$ has \emph{high} degree (or \emph{is high}) if $d(v) > \sqrt{\om}$, and that it has \emph{medium} degree (or \emph{is medium}) if $d(v) \in (\ogamma,\sqrt{\om}]$.
We refer to copies that contain only low vertices as \emph{low} copies, and to all other copies as \emph{mixed}.

\paragraph{Sampling low copies.}
The first, and simpler, procedure, which samples low copies, works as follows. It starts by sampling a uniformly distributed vertex $v_1$ in the graph. 
Conditioned on $v_1$ being low, it performs a ``uniform random walk'' of length $k-1$ restricted to low vertices. 
By this we mean that  in each step $i$, each  neighbor of the current vertex $v_i$ (which is ensured to be a low vertex) is chosen with probability exactly $\frac{1}{\ogamma}$ (indeed if $d(v_i) < \ogamma$, then it is possible that no neighbor is selected).  
If at any point  the walk encounters a high or medium vertex or fails to sample a neighbor, the procedure fails. 
If a path $(v_1,\dots,v_k)$  over low vertices was successfully obtained, 
then the procedure perform a single pair query 
between $v_k$ and $v_1$ to check whether this results in a (low) copy of $\Cyk$. If a copy is obtained, then the procedure returns it   with probability $\frac{1}{2k}$, and otherwise the procedure fails.

For each specific low copy, the probability that it is selected is
$\plk = 2k \cdot \frac{1}{n}\cdot \left(\frac{1}{\ogamma}\right)^{k-1}\cdot \frac{1}{2k} = \frac{1}{n\cdot \ogamma^{k-1}} = \frac{1}{n\cdot \onC^{1-1/k}}$,  where the factor of $2k$ 
accounts for the cyclic orderings of the vertices.

\ifnum\stoc=0
\begin{figure}
\centerline{\mbox{
\includegraphics[width=0.35\linewidth]{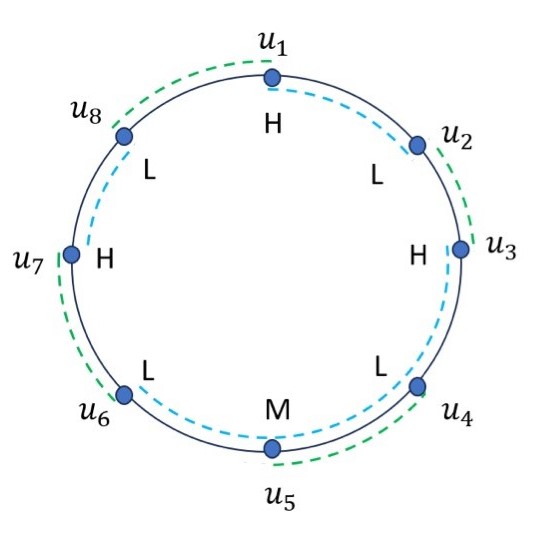}}}
    \caption{\small The letters `H', `M' and `L' signify whether a vertex is high, medium or low, respectively. The cycle in the figure can be covered by several sequences of paths. One example is the sequence of paths, depicted in green (outside the cycle), $\pi_1 = (u_8, u_1), \pi_2 = (u_2, u_3), \pi_3 = (u_4, u_5), \pi_4 = (u_6, u_7)$. Another example is the sequence of paths, depicted in turquoise (inside the cycle), $\pi'_1 = (u_1, u_2), \pi'_2 = (u_3, u_4, u_5, u_6), \pi'_3 = (u_7, u_8)$. 
       }
    \label{fig:cover-intro}
\end{figure}
\else
\begin{figure}
\includegraphics[width=0.45\linewidth]{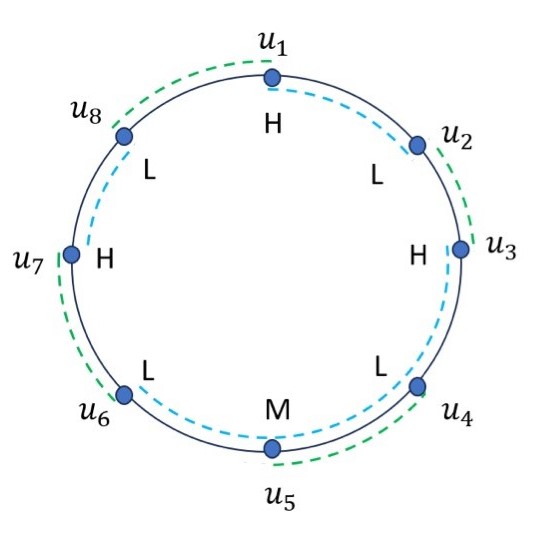}
    \Description{}
    \caption{\small The letters `H', `M' and `L' signify whether a vertex is high, medium or low, respectively. The cycle in the figure can be covered by several sequences of paths. One example is the sequence of paths, depicted in green (outside the cycle), $\pi_1 = (u_8, u_1), \pi_2 = (u_2, u_3), \pi_3 = (u_4, u_5), \pi_4 = (u_6, u_7)$. Another example is the sequence of paths, depicted in turquoise (inside the cycle), $\pi'_1 = (u_1, u_2), \pi'_2 = (u_3, u_4, u_5, u_6), \pi'_3 = (u_7, u_8)$. 
       }          
\label{fig:cover-intro}
\end{figure}
\fi

\paragraph{Sampling mixed copies.}\label{subsubsec:med-high}
Recall that we sample mixed copies by considering sequences of paths that together cover all vertices on a copy. We sample each path separately, and then check if we can ``glue'' the paths together to obtain a single cycle. Each path (which covers between $1$ and $k$ vertices) starts with a medium or high vertex, and all other vertices on the path (if they exist) are either medium or low. Furthermore, if the path contains a single vertex, then it must be high. See Figure~\ref{fig:cover-intro} for an illustration. We emphasize that these paths are defined with respect to copies of the motif $\F$ in $G$, depending on the degrees of vertices in such copies (rather then with respect to $\F$ itself).

These constraints on the paths give us the following. 
(1) The lower bound (of $\ogamma$) on the degree of the first vertex of each path allows us to sample this vertex as well as the first edge on the path (if such exists) at a relatively low cost. \talya{this is a bit vague. why not say we can sample the first edge at unit cost?}
(2) The fact that all other vertices on each path have degree at most $\sqrt{\om}$, 
allows  to take uniform random steps with probability $\frac{1}{\sqrt{\om}}$   starting from the second vertex on the path. 

%

More precisely, the fact that each path starts with a vertex whose degree is lower bounded by $\ogamma$ allows us to take advantage of the following approach, which was applied in the past in the context of sublinear counting and sampling~\cite{cliquesSicomp, EMR_multiple_edges}. 
By sampling an auxiliary (multi-)set of vertices of size roughly $\frac{n}{\ogamma}$, 
it is possible  to construct a data structure that (with high probability) can be used to sample \emph{at unit cost} each vertex $v$ having degree greater than $\ogamma$ with probability approximately $\frac{d(v)}{\om}$. This has two related implications. 

The first is that we can sample edges incident to medium and high vertices, each with probability $\frac{1}{\om}$. This is done by first sampling a medium or high  vertex using the data structure, and then querying for one of its neighbors uniformly at random. Thus each such (ordered) edge $(v,u)$ is sampled with probability $\frac{d(v)}{\om}\cdot \frac{1}{d(v)}= \frac{1}{\om}$. 
The second is that we can sample high vertices with 
probability $\frac{1}{\sqrt{\om}}$ by first sampling a high vertex $v$ with probability proportional to its degree and then ``keeping'' it with probability $\frac{\sqrt{\om}}{d(v)}$, so that overall, each 
high vertex is sampled with probability $\frac{d(v)}{\om}\cdot \frac{\sqrt{\om}}{d(v)}=\frac{1}{\sqrt{\om}}$. (This only works for high vertices since we need that $d(v)\geq \sqrt{\om}$ to have $\frac{\sqrt{\om}}{d(v)}\leq 1$.)

Putting the above together, for each fixed sequence of paths that covers a specific mixed copy as described above, the probability that we obtain this sequence is (approximately) $\frac{1}{\sqrt{\om}^k}$: Each path in the sequence that consists of a single (high) vertex is obtained with probability $\frac{1}{\sqrt{\om}}\,$; For each longer path, the first edge (pair of vertices) is obtained with probability $\frac{1}{\om}$, and each subsequent vertex on the path is obtained with probability $\frac{1}{\sqrt{\om}}$. Since the sum of the lengths of the paths is $k$, we get the claimed sampling probability.

Up until now we considered a fixed sequence of paths that covers a specific copy. Since such a copy  is potentially covered by several different sequences of varying lengths, as a last step after obtaining a copy, we return it with probability inverse proportional to the number of sequences of paths that cover it. (This is formalized in Definition~\ref{def:fitting} and Algorithm~\ref{alg:sampHMCopies}.)

\subsubsection{Attempted sampling procedures for a general Hamiltonian motif $\F$ }\label{subsubsec:Ham}
The sampling procedures (for low and for mixed copies) are modified as follows for sampling any Hamiltonian motif $\F$ over $k$ vertices. If a copy $\Cy$ of $\Cyk$ is selected by the procedure, rather than returning it (with the appropriate probability), the procedure first performs all pair queries to obtain the induced subgraph, denoted $\Q$. If $\Q$ contains a copy of $\F$ that in turn includes $\Cy$, then the procedure selects one of these copies with equal probability, where this probability depends on $\F$ and is lower bounded by $1/k!$. 

If a copy of $\F$ is selected, then, similarly to the case when $\F=\Cyk$,  the selected copy is returned with probability inverse proportional to the number of ways it could have been selected by the algorithm. In particular, recall that in the case that $\F=\Cyk$, this number is determined by the number of sequences of paths that cover the cycle. Now, for a general Hamiltonian $\F$,  we also need to take into account the number of Hamiltonian cycles of the copy, since each can give rise to the selection of the copy.
The counting algorithm uses the sampling procedures in exactly the same way as for the special case of $\F=\Cyk$.


\subsubsection{The importance of Hamiltonian cycles}

Our algorithm relies on the fact that $\F$ is Hamiltonian. This property of $\F$ does not come into play when sampling low copies. Indeed, we can sample copies of any (connected) motif $\F$ over $k$ vertices that solely contain low-degree vertices, where each copy is returned with equal probability, which is roughly $\frac{1}{n\cdot \onF^{1-1/k}}$ (as in the case of a Hamiltonian $\F$). 

However, the Hamiltonicity of $\F$ becomes  crucial for our algorithm when turning to copies that contain at least one high-degree vertex. 
The reason is that we rely on the ability to ``cover'' all vertices in such a copy by a sequence of paths, where each path starts with a vertex that is either medium or high, and all other vertices are either medium or low. 
This is always possible when the copy contains a Hamiltonian cycle, but not in general.\footnote{Consider for example a path over four vertices where  one of the medium vertices is high and all other vertices are low.}




\subsection{From attempted sampling procedures to approximate counting}\label{subsubsec:intro-approx-count}
Building on  uniform (or almost uniform) attempted samplings procedures to obtain an approximate-counting algorithm
is fairly standard, but since there are some subtle details specific to our case, we describe it here in order to provide a full picture.
For simplicity of the presentation, here too we restrict our attention to the case that $\F$ is a $k$-cycle. As before, we assume we have a constant factor estimate $\om$ of the number of (ordered) edges and a constant factor estimate $\onF$ of $\countF$.

Recall that we have two attempted sampling procedures for $k$-cycles. The first, given $n$ and $\ogamma = \onF^{1/k}$,  returns each low copy with probability
$\frac{2k}{\Bl}$ where 
$\Bl = n\cdot \ogamma^{k-1}$.
The second, given $\om$, returns each mixed copy with probability (approximately) $\frac{2k}{\Bmh}$ where $\Bmh = {\om^{k/2}}$. We can devise a single combined attempted sampler that returns each copy (of any one of the two types) with probability (approximately) $\frac{1}{B}$, where $B = \Bl+\Bmh$  by invoking the first procedure with probability $\frac{\Bl}{B}$ and the second with probability $\frac{\Bmh}{B}$. 

Recall that the
 second procedure (for mixed copies)  actually works under the assumption that it is provided with a data structure that allows  to sample (in unit time) vertices having degree greater than $\ogamma$ with probability approximately proportional to their degree.
Hence, the approximate-counting algorithm starts by constructing such a data structure, which can be done  in time  $O^*\Big(\frac{n}{\ogamma}\Big)$.

Next, the algorithm invokes the combined attempted sampling procedure $t=\Theta\Big(\frac{B}{\onF\cdot \eps^2} \Big)$ times and sets 
$t_1$ to be the number of times a copy is returned. Observe that conditioned on the success of the construction of the data structure, in each invocation, the probability of obtaining a copy is $\frac{\countF}{B}$, where these events are independent. 
By applying the multiplicative Chernoff bound (using $\onF = \Theta(\countF)$) we get that 
$\frac{t_1}{t} \cdot B$ is a
$(1\pm \eps)$-estimate of $\countF$ with high constant probability.
Recalling that $\ogamma = \onF^{1/k}$, the resulting complexity of the algorithm is $O^*\Big(\frac{n}{\ogamma} + \frac{n\cdot \ogamma^{k-1}+ \om^{k/2}}{\onF}\Big) = O^*\Big(\frac{n}{\onF^{1/k}} + \frac{\om^{k/2}}{\onF}\Big)$.

\subsection{Pointwise uniform sampling}\label{subsec:intro-approximate-sampling}
Given our algorithm for approximate counting and the combined attempted sampling procedure discussed in Section~\ref{subsubsec:intro-approx-count},
we 
establish the result of Theorem~\ref{thm:sampling} in a fairly
straightforward manner. 
We first compute a constant-factor estimate $\om$ of the number of (ordered) edges (using one of the aforementioned algorithms~\cite{Feige-Avg,GR08}) and a constant-factor estimate $\onF$ of $\countF$ using our approximate-counting algorithm.
Given these estimates we construct a data structure for sampling vertices with degree greater than $\ogamma = \onF^{1/k}$.
The above constitutes the preprocessing part of the final sampling algorithm (referred to in Theorem~\ref{thm:sampling}). 

Now, to obtain a copy distributed pointwise-close to uniform, 
we implement a Las-Vegas algorithm that simply invokes the combined attempted sampler until a copy is returned. Conditioned on the event that the preprocessing step is successful, the expected query complexity and running time of this nearly uniform sampler 
are as stated in Theorem~\ref{thm:sampling}.

\subsection{Simplified attempted sampling procedures for $k$-cliques and $k$-stars}

\subsubsection{Procedures for $k$-cliques}
In the case of $k$-cliques, we can exploit the fact that they have diameter-1 to simplify the attempted sampling procedures. In particular, we no longer need to find Hamiltonian cycles of copies, so we can do without obtaining covering paths. Instead, we have three attempted procedures for sampling copies, depending on the minimum degree vertex in the copy. A copy
is classified as low, medium or high, according to its min degree vertex.

The procedure for sampling a low copy  selects a vertex uniformly at random, and if it is low, the procedure  selects $k-1$ of its neighbors, each with probability $\frac{1}{\ogamma}$, and checks whether it obtained a clique. The procedure for sampling a medium copy samples a vertex with degree greater than $\ogamma$ as described in Section~\ref{subsubsec:med-high}, and if it is medium, the procedure selects one of its neighbors uniformly at random, and $k-2$ additional neighbors each with probability $\frac{1}{\sqrt{\om}}$. It too checks if it obtained a clique. Finally, the procedure for sampling a high copy selects $k$ high vertices, each with probability $\frac{1}{\sqrt{\om}}$ (also as described in Section~\ref{subsubsec:med-high}), and checks if they form a clique. 
Thus, each low copy is sampled with probability $\frac{1}{n\cdot \ogamma^{k-1}}$, each medium copy with probability (approximately) $\frac{1}{\om}\cdot\frac{1}{\sqrt{\om}^{k-2}} = \frac{1}{\om^{k/2}}$, and each high copy with probability (approximately) $\frac{1}{\sqrt{\om}^{k}}  = \frac{1}{\om^{k/2}}$. 

\subsubsection{Procedures for $k$-stars}

For the case of $k$-stars we also present two attempted samplers, for two types of copies, where here the types of copies are determined according to the degree of the central star vertex. Sampling low copies, where the central vertex is low, is done identically to sampling low $k$-cliques (without the clique verification step).
Sampling non-low copies, where the central vertex is either medium or high, is slightly different than previous samplers. 
Here we rely on the observation that in a graph with $\countS$ stars, the maximum degree  is  (up to normalization factors) at most $\countS^{1/(k-1)}$. Therefore, given an estimate on the number of stars, the procedure sets an estimate $\odmax$ on the maximum degree in the graph, and continues at follows. It samples a
medium-high edge $(u,v)$ and then samples  $k-2$ additional neighbors of $u$ such that each is sampled with probability $1/\odmax$. Overall, each non-low copy is sampled with probability $\frac{1}{\om\cdot \odmax^{k-2}}$.
We note that the latter sampler is very similar to the one described in~\cite{BER}.

\subsection{Related Work}\label{sec:related}

Most of the related works for sublinear 
approximate counting and sampling of motifs, both in %
the standard  model and in  the augmented model, were already mentioned earlier in the introduction, so we do not 
cover them here.

\paragraph{Sublinear approximate counting and sampling in bounded arboricity graphs.}
The tasks of approximate counting and sampling in the standard model were also considered for the special class of graphs with bounded arboricity. For this class of graphs there are refined results for approximately counting  edges and stars~\cite{ERS19},  approximately counting  $k$-cliques~\cite{ERS20} and pointwise close to uniform sampling of edges~\cite{ERR19} and $k$-cliques~\cite{ERR_sampling_cliques}.

\paragraph{Approximate subgraph counting 
using fast matrix multiplication.}
In~\cite{tetek_triangles}, T\v{e}tek  considers the problem of approximating the number of triangles, denoted $\nT$. When $\nT = O(\sqrt m)$ his result improved over the previous state of the art of ~\cite{ELRS}.  
We note that in this regime, sublinear complexity is unattainable because $\Omega(m)$ queries are necessary. He  presents two algorithms that rely on fast matrix multiplication. One is suitable for dense graphs and runs in time
$O^*\Big(\frac{n^\omega}{\nT^{\omega-2}}\Big)$. The other is suitable for sparse graphs and runs in time
$O^*\Big(\frac{m^{2\omega/(\omega+1)}}{\nT^{2(\omega-1)/(\omega+1)}}\Big)$ (where $\omega$ is the matrix-multiplication exponent).
Censor-Hillel, Even and Williams~\cite{virginia_cycles} simplify and generalize the above result to any constant length cycle in 
directed graphs, and also for odd length cycles in  undirected graphs. For triangles %
they obtain the upper bounds $O^*\Big(n^2+\frac{n^\omega}{\nT^{(\omega-2)/(1-\alpha)}}\Big)$  and  $O^*\Big(m+\frac{m^{2\omega/(\omega+1)}}{\nT^{2(\omega-1)/(\omega+1)+\alpha(\omega-2)/(1-\alpha)(\omega+1)}}\Big)$
where $\alpha \geq 0.32$.
For cycles of size $k>3$, their algorithm has running time
 $O^*(MM(n,n/\countC^{1/(k-2)},n))$ where $MM$ denotes the time to multiply  an $n \times 
n/\countC^{1/(k-2)}$ matrix by an $n/\countC^{1/(k-2)} \times n$ matrix (recall that $\countC$ denotes the number of  $k$-cycles in the graph). They also sate a related fine-grained hardness hypothesis,      under which, their algorithm is optimal for dense graphs in the non-sublinear regime.

\paragraph{Other sublinear-time query models.}
In the sublinear regime, 
there are query models that allow for different types of set queries~\cite{BGM19a, beame2020edge, Chen-is-edges,bhattacharya2021triangle,DLM22, Bishnu-triangles} or for different types of access to the neighborhoods of vertices~\cite{TT_edges}. 

\paragraph{Lower bounds.}
In~\cite{ER18_lbs}, Eden and Rosenbaum presented a framework for proving lower bounds on subgraph counting via reductions from communication complexity. Their results simplified on previous bounds that were obtained by first principles. In~\cite{assadi2022asymptotically}, Assadi and Nguyen studied the dependencies on the approximation and error parameters and proved asymptotically optimal
lower bounds for the  problem of approximately counting triangles.

\subsection{Organization}
\ifnum\stoc=1
Following preliminaries in Section~\ref{sec:prelim}, in Section~\ref{sec:hamiltonian_motifs} we describe our attempted samplers and approximate counting algorithm for Hamiltonian motifs. In Section~\ref{sec:simple-samplers-cliques-stars} we present the simplified attempted samplers for cliques and stars.
Our pointwise uniform sampling algorithms and all omitted proofs appear in the full version~\cite{ELRR25-arxiv}.
\else
Following preliminaries in Section~\ref{sec:prelim}, in Section~\ref{sec:hamiltonian_motifs} we describe our attempted samplers and approximate counting algorithm for Hamiltonian motifs, as well as the pointwise-close-to-uniform sampler. In Section~\ref{sec:simple-samplers-cliques-stars} we present the simplified attempted samplers for cliques and stars.
All notations introduced throughout the paper appear in 
Table~\ref{tab:notation} -- see Appendix~\ref{app:notation}.
\ForFuture{ 
(0) should we list the appendices as well?
(1) check if refer to sampler ok everywhere. (2) Flow of long version given changes in short. (3) space before/after algorithms.
(4) Discuss option of 3 algorithms (5) Future: general counting and sampling schemes based on attempted samplers}
\fi

\section{Preliminaries}\label{sec:prelim}

Let $[r]$ denote the set $\{1,\ldots,r\}$. We use the notation $x \in (1\pm \eps)\cdot y$ as a shorthand for $(1-\eps)\cdot y \leq x \leq (1+\eps)\cdot y$.

Let $G = (V,E)$ be an undirected simple graph over $n$ vertices.
For a vertex $v$ let $d(v)$  denote the degree of $v$ and let $\Gamma(v)$ denote its set of neighbors. 
Let $\davg$ denote the \textit{average} degree in $G$.
 We use $m$ to denote the number of \emph{ordered} edges $(u,v)$ in $G$ so that $m = 2|E|$ and
 $m = \sum_v d(v) = n\cdot \davg$.

Let $\Cyk$  denote the cycle over $k$ vertices. When referring to a subgraph of $G$ that is a cycle, we use
the notation $(v_1,\dots,v_k,v_1)$ (where $\big\{\{v_i,v_{i+1}\}:i\in [k-1]\big\}\cup \big\{\{v_k,v_1\}\big\} \subseteq E$).

We consider the standard  graph query model~\cite{PR,KKR04}, where the algorithm can perform the following types of queries:
\ifnum\stoc=0
\begin{inparaenum}[(1)] 
\item For any  $v\in V$, query for $v$'s degree;
\item For any vertex $v\in V$ and index $i$, 
query for $v$'s $i\th$ neighbor (if $i> d(v)$, then a special symbol is returned);
\item For any pair of vertices $u,v$, query whether $\{u,v\}\in E$.
\end{inparaenum}
\else
(1) For any  $v\in V$, query for $v$'s degree;
(2) For any vertex $v\in V$ and index $i$, 
query for $v$'s $i\th$ neighbor (if $i> d(v)$, then a special symbol is returned);
(3) For any pair of vertices $u,v$, query whether $\{u,v\}\in E$.
\fi

 A \emph{motif} $\F$ is a graph over $k$ vertices (which is known and considered to be small).

\begin{definition}[Copies of a motif $\F$]\label{def:copy}
 For a graph $G$ and a motif $\F$,  we say that a subgraph $\f$ of $G$   is a \textsf{copy} of $\F$ in $G$ if $\f$ is isomorphic to $\F$.
 
Let $\countF(G)$  denote the number of copies of $\F$ in $G$, and when $G$ is clear from the context, we use the shorthand $\countF$.
\end{definition}

\begin{definition}[Hamiltonian cycles of a copy]\label{def:Ham}
For a copy $\f$ of $\F$ in $G$, let $\ham(\f)$ denote the set of different copies of $\Cyk$ in $\f$ (i.e., Hamiltonian cycles of $\f$). 
~Since $|\ham(\f)|$ is the same for all copies $\f$ of $\F$, we denote this value by $\hF$. 
\end{definition}

\begin{definition}[Number of copies in a subgraph that contain a given cycle]
\label{def:nF}
 For  a graph $Q$, a motif $\F$ and a cycle $\Cy\in \ham(Q)$, each with $k = |V(\F)|$ vertices, we let $\countF(Q, \Cy)$ denote the number of different  copies of $\F$ in $Q$ that include the cycle $\Cy$.     

We note that for a fixed $\F$, $\countF(Q, \Cy)$ is maximized when $Q$ is a clique over $k= |V(\F)|$ vertices and $\Cy$ is some Hamiltonian cycle of $Q$. 
We denote this value by $\kappaF$.\footnote{\label{foot:kappaF} Observe that $\kappaF$ ranges between $1$ and $k!$.}
\end{definition}
For an illustration of Definition~\ref{def:nF}, see Figure~\ref{fig:copies}.

\ifnum\stoc=0
All notations introduced in this section as well as the following sections can be found in Table~\ref{tab:notation} -- see Appendix~\ref{app:notation}.
\fi

\ifnum\stoc=0
\begin{figure}
    \centering
\includegraphics[width=0.7\linewidth]
  {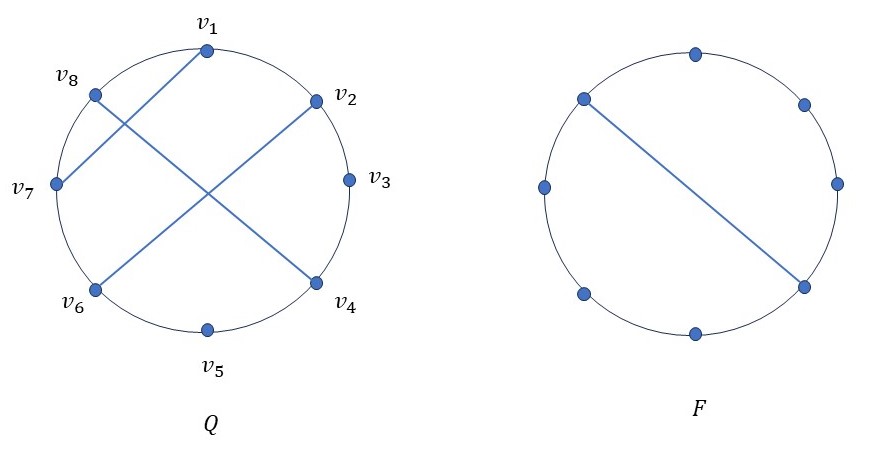}
    \caption{There are two different copies of $\F$ in $Q$ that include the Hamiltonian cycle $(v_1, v_2, \ldots v_8,v_1)$: one that uses the cord $\{v_2, v_6\}$ and one that uses the cord $\{v_4, v_8\}$. Denoting this cycle by $\Cy$, we have that $\countF(Q, \Cy) = 2$. }
    \label{fig:copies}
\end{figure}
\else
\begin{figure}
\includegraphics[width=0.9\linewidth]
  {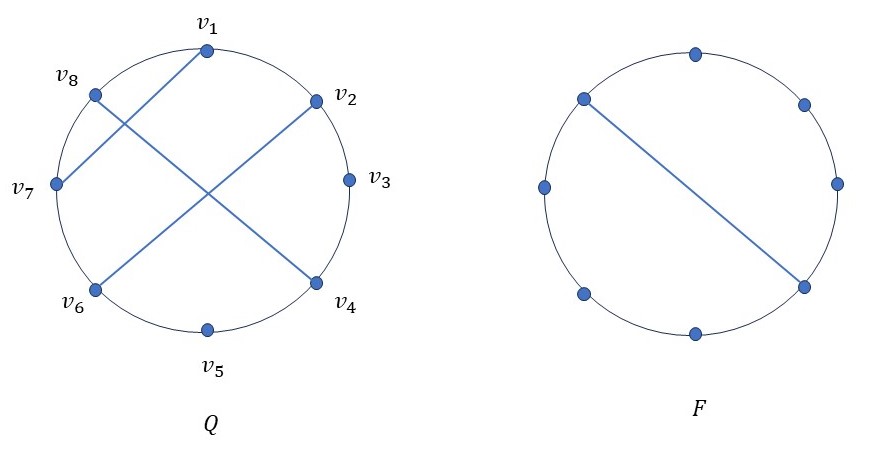}
  \Description{}
    \caption{There are two different copies of $\F$ in $Q$ that include the Hamiltonian cycle $(v_1, v_2, \ldots v_8,v_1)$: one that uses the cord $\{v_2, v_6\}$ and one that uses the cord $\{v_4, v_8\}$. Denoting this cycle by $\Cy$, we have that $\countF(Q, \Cy) = 2$. }
    \label{fig:copies}
\end{figure}
\fi

\ifnum\stoc=0
\section{Approximately counting and 
sampling Hamiltonian motifs}
\else
\section{Approximately counting  Hamiltonian motifs}
\fi
\label{sec:hamiltonian_motifs}



As discussed in the introduction, our sublinear approximate counting and nearly uniform sampling  algorithms build on several subroutines, which are discussed in the next subsections. \talya{the previous sentence is not saying much} We start with two procedures in Section~\ref{subsec:preprocess-and-samp-high}. The first, \preprocess, uses a sample of vertices from $V$  to construct a data structure that allows the second procedure, \sampHighVertex, to sample vertices of sufficiently high degree with probability roughly proportional to their degree. In Section~\ref{subsec:samp-low-copy} we present and analyze the procedure  \sampLowCopies\ for sampling low copies (i.e., copies in which all vertices have degree upper bounded by a certain degree threshold 
$\ogamma$). In Section~\ref{subsec:samp-mixed-copy} we present and analyze the procedure \sampHMCopies\ for sampling mixed copies (i.e., copies that contain at least one vertex with degree above the threshold $\ogamma$). The latter procedure invokes the aforementioned procedure \sampHighVertex. The two procedures are combined in Section~\ref{subsec:sample-copy}.
\ifnum\stoc=1
Our approximate counting algorithm, which uses all these procedures, is provided in Section~\ref{subsec:approx-count}.
\else
Our approximate counting and nearly uniform sampling algorithms, which use all these procedures,  are provided in Sections~\ref{subsec:approx-count} and~\ref{sec:uniform-sampling}, respectively.
\fi

For the sake of conciseness, we sometimes do not explicitly say that the algorithm performs  certain queries, but only refer to the outcome of the queries. For example, we might refer to $d(v)$ for a vertex $v$, without explicitly stating that the algorithm performs a degree query to obtain $d(v)$, and we might refer to the subgraph induced by a set of vertices without explicitly saying that the algorithm performs pair queries to obtain this subgraph. We shall of course take this into account when bounding the query complexity of the algorithms.

\subsection{A degrees-typical data structure supporting sampling  vertices }
\label{subsec:preprocess-and-samp-high}

The following definition and two lemmas are  adapted from~\cite{cliquesSicomp, EMR_multiple_edges}. 
The proofs of the lemmas (including the pseudo-code of the procedures that they refer to) can be found in 
\ifnum\stoc=0
Appendix~\ref{app:samp-high}.
\else
the full version.
\fi

\begin{definition}[An $(\oeps,\ogamma,\om)$-degrees-typical multiset and data structure]
\label{def:good-DS}
For a graph $G=(V,E)$ over $n$ vertices and parameters $\oeps \in (0,\frac{1}{2})$ and $\ogamma,\om \in \mathbb{N}^+$, 
    we say that a multiset  $S$ of vertices from $V$ is \textsf{$(\oeps,\ogamma,\om)$-degrees-typical} with respect to $G$, if the following two conditions hold:
    \begin{enumerate}

        \item For every $v\in V$ with $d(v)> \ogamma$, we have that $d_S(v)
        \in (1\pm\oeps)\frac{d(v)}{n}$.

        \item Letting $m(S)=\sum_{u\in S}d(u)$ and $\od = \frac{\om}{n}$,
        we have that $m(S)\leq 2\cdot |S|\cdot \od$. 
    \end{enumerate}
   A data structure $\calD$ is \textsf{$(\oeps,\ogamma,\om)$-degrees-typical} with respect to $G$, if for some multiset $S$ that is
   $(\oeps,\ogamma,\om)$-degrees-typical with respect to $G$, the data structure $\calD$  supports sampling each vertex $u\in S$ with probability $\frac{d(u)}{2\od|S|}$ in constant time. 
   If $G$ is clear from the context, then we shall just say that $\calD$ is $(\oeps,\ogamma,\om)$-degrees-typical.
\end{definition}

\begin{lemma}[Constructing a degrees-typical data structure] 
\label{lem:preproc} 
The  procedure
\preprocess\ is given query access to a graph $G$ and parameters  $n$, $\oeps \in (0,\frac{1}{2})$, $\delta\in(0,1)$ and $\ogamma,\om \in \mathbb{N}^+$, 
and satisfies the following.
\begin{itemize}

 \item 
 If 
 $\om \geq m$
 and 
 $\ogamma\leq \sqrt{\om}$,
 then 
\preprocess$(n,\oeps,\delta,\ogamma,\om)$ returns an $(\oeps,\ogamma,\om)$-degrees-typical data structure $\calD$ with probability at least $1-\delta$. 

\item The query complexity and  running time of \preprocess\ are $O\left(\frac{n}{\ogamma}\cdot \frac{\log(n/\delta)\cdot \log(1/\delta)}{\oeps^2}\right)$. 
   \end{itemize}
\end{lemma}

\begin{lemma}[Sampling medium-high vertices]\label{lem:sampHeavyVertex}
Let $G$ be a graph 
and let  $\oeps \in (0,\frac{1}{2})$ and $\ogamma,\om \in \mathbb{N}^+$. 
The procedure \sampHighVertex\ is given query access to $G$ and
 an $(\oeps,\ogamma,\om)$-degrees-typical data structure $\calD$ (w.r.t. $G$), and satisfies the following. 
 \begin{itemize}

 \item 
 Each invocation of \sampHighVertex$(\calD,\ogamma)$ either outputs a vertex $v$ or fails. The distribution on output vertices is such that every  vertex $v$ for which
$d(v)> \ogamma$ is returned with probability in $(1\pm\oeps)\frac{d(v)}{2\om}$ and every vertex $v$ with $d(v)\leq\ogamma$ is returned with probability $0$. 

\sloppy
\item The query complexity and running time  of (a single invocation of) \sampHighVertex{} 
are $O(1)$. 
\end{itemize}
\end{lemma}

\subsection{Sampling low copies}\label{subsec:samp-low-copy}
Recall that a low copy of $\F$ in $G$ is a copy 
whose vertices all have degree at most $\ogamma$ (for a given parameter $\ogamma$).
 The sampling algorithm begins by selecting a vertex uniformly at random. If its degree is low, the algorithm performs a uniform random walk of 
$k-1$ steps, restricted to low-degree vertices.
If the walk succeeds in obtaining a path over  $k$ low-degree vertices, then the algorithm checks whether they close a cycle, and whether this is a Hamiltonian cycle of some copy of $\F$, in which case a copy is returned with appropriate probability.
 Precise details follow.

Recall that the notations $\hF$ and $\kappaF$ were introduced in the preliminaries (See Definitions~\ref{def:Ham} and~\ref{def:nF}).


\alg{alg:sampLowCopies}{
	{\bf \sampLowCopies}$\;(\F,\ogamma)$ 
	\begin{enumerate}[itemsep=.2em, leftmargin=*]
		\item Sample a vertex $v_1 \in V$ uniformly at random. If $d(v_1) > \ogamma$, then \textbf{return} \emph{fail}.\label{step:sample-uniform}
    \item For $i=1$ to $k-1$:
    \begin{enumerate}[itemsep=.2em]
          \item Choose an index $j\in [\ogamma]$ u.a.r.
          If $j > d(v_i)$, then \textbf{return} \emph{fail}. Otherwise,
          query for $v_i$'s $j\th$ neighbor and let  $v_{i+1}$ denote the  neighbor returned. 
          \label{step:sample-neighbor-light}
          \item If $d(v_{i+1}) > \ogamma$, then \textbf{return} \emph{fail}. 
    \end{enumerate}
      \item If $v_i = v_{i'}$ for some $1 \leq i < i' \leq k$ or $\{v_k, v_1\}\notin E$,  then \textbf{return} \emph{fail}. \label{step:check-low}   \newcomment{Verify that no vertex is sampled twice and that the path forms a cycle.}
      \item \label{step:select-low-copy} Let $\Cy$ denote the cycle $(v_1, \ldots v_k, v_1)$, and  let $\Q$ be the subgraph induced by the set of vertices  $\{v_1,\dots,v_k\}$. If $\Q$ contains at least one copy of $\F$, then either select one such copy $\f$ or \textbf{return} \emph{fail}, where for each copy, the probability that it is selected is $1/\kappaF$. 
      \newcomment{Recall: $\kappaF$ is defined in Definition~\ref{def:nF}}
      \item Return $\f$ with probability $\frac{1}{\hF\cdot 2k}$ and with probability $1-\frac{1}{\hF\cdot 2k}$. 
      \textbf{return} \emph{fail}. \newcomment{Recall: $\hF$ is defined in Definition~\ref{def:Ham}} \label{step:check-iso}  
      \label{step:return-low-copy}
	\end{enumerate}
}

\begin{lemma}\label{lem:samp-low}
    Algorithm \sampLowCopies$(\F,\ogamma)$ either returns a  copy of $\F$ in $G$ such that all vertices in the copy have degree at most $\ogamma$ or fails. Furthermore, each  such copy is returned  with probability $\frac{1}{n\cdot \ogamma^{k-1} \cdot \kappaF}$.
    The query complexity of the algorithm is $O(k^2)$, 
    and the running time is upper bounded by $O(k!\cdot k^2)$.
\end{lemma}

We note that, depending on $\F$,  both the query complexity and the running time of the algorithm may be smaller than what is stated in the lemma.
For example if $\F = \Cyk$, then there is no need to obtain the induced subgraph $\Q$, so that the  query complexity and running time are both $O(k)$. 


\ifnum\stoc=0
\begin{proof}
The algorithm performs at most $k$ degree queries (one on each of the vertices $v_i$ for $i\in [k]$), at most $k-1$ neighbor queries (one on each of the vertices $v_i$ for $i\in [k-1]$), and less than ${k \choose 2}$ pair queries (to check whether $\{v_k,v_1\}\in E$ and to obtain the induced subgraph $\Q$). The upper bound on the query complexity of the algorithm follows. The running time is dominated by  Step~\ref{step:check-iso},  which can be implemented by enumerating over all possible mappings between the vertices of $\F$ and the vertices of $\Q$, and for each checking whether we get a subgraph of $\Q$ that is isomorphic to $\F$.
Hence the running time of this step is upper bounded by $O(k!\cdot k^2)$, and the upper bound on the  running time follows.

Since the algorithm checks that $d(v_i)\leq \ogamma$ for all $i\in[k]$ (and otherwise fails), if it returns a copy of $\F$ in $G$, then all vertices in the copy have degree at most $\ogamma$. It remains to bound the probability that any such specific copy is returned.

Consider a fixed copy $\f^*$ of $\F$ in $G$ where all vertices in the copy have degree at most $\ogamma$. Recall that $v_i$ is the $i\th$ vertex selected by the algorithm.
The copy $\f^*$ is returned if the following events occur:
\begin{enumerate}
    \item \sloppy 
 For one of the $\hF = |\ham(\f^*)|$ copies of $\Cyk$ in $\f^*$, and for one of the $2k$ possible orderings of the vertices on the cycle,\footnote{\label{foot:orderings} The $2k$ orderings are determined by the starting vertex ($k$ options) and the directions of the cycle ($2$ options).} the $i\th$ sampled vertex $v_i$ hits the $i\th$ vertex in the ordering.
    
    Observe that for a fixed copy of $\Cyk$ in $\f^*$, i.e., a Hamiltonian cycle of $\f^*$, and a fixed ordering of the vertices on the cycle, this event occurs with probability $\frac{1}{n}\cdot \frac{1}{\ogamma^{k-1}}$, as $v_1$ ``hits'' the first vertex in the ordering with probability $\frac{1}{n}$ and each subsequent step hits the next vertex in the ordering with probability $\frac{1}{\ogamma}$. 
    \item The copy $\f^*$ is selected  in Step~\ref{step:select-low-copy} (conditioned on the first event occurring). This event occurs with probability $\frac{1}{\kappaF}$.
    \item The copy $\f^*$ is returned in Step~\ref{step:return-low-copy} (conditioned on the first two events occurring). This event occurs with probability $\frac{1}{\hF\cdot 2k}$.
\end{enumerate}
Therefore, the overall probability that $\f$ is returned equals
$$\hF \cdot 2k \cdot \frac{1}{n}\cdot \frac{1}{\ogamma^{k-1}} \cdot \frac{1}{\kappaF}\cdot \frac{1}{\hF\cdot 2k} = \frac{1}{n\cdot \ogamma^{k-1}\cdot \kappaF }\;,$$
as claimed.
\end{proof}
\fi

\subsection{Sampling mixed copies} \label{subsec:samp-mixed-copy}
A mixed copy of $\F$ in $G$ is a copy that contains at least one vertex with degree exceeding $\ogamma$.
Our algorithm for sampling mixed copies tries to obtain a Hamiltonian cycle of such a copy by constructing a sequence of paths that obey certain constraints on the degrees of the vertices belonging to the paths.

In order to formalize the above and provide full details, we introduce the next two definitions.
For an illustration of the notations introduced in these definitions, 
see Figure~\ref{fig:fits}.

\begin{definition}[paths cover] \label{def:cover}
 Let 
 $\Cy = (u_1,\dots,u_k,u_1)$ be a copy of $\Cyk$ in $G$.
 For positive integers $\ogamma$ and $\om$
 such that $\ogamma \leq \sqrt{\om}$,
 we say that a sequence of $r \geq 1$ paths 
$\langle \pi_1,\dots,\pi_r \rangle$
 \textsf{covers} $\Cy$ with respect to $\ogamma$ and $\om$, if the following conditions hold.
 \begin{enumerate}
     \item Every vertex in $\Cy$ belongs to exactly one of the $r$ paths, and every edge in each of the paths belongs to $\Cy$.
     \item For every $q\in [r-1]$ there is an edge between the last vertex in $\pi_q$ and the first vertex in $\pi_{q+1}$, and there is also an edge between the last vertex in $\pi_r$ and the first vertex in $\pi_1$.
     \item For every $q \in [r]$, 
 the first vertex in $\pi_q$ has degree greater than $\ogamma$, and if $|\pi_q|=1$, then it has degree greater than $\sqrt{\om}$. Every other vertex on the path
has degree at most $\sqrt{\om}$.

 \end{enumerate}
  \end{definition}

Observe that every copy $\Cy = (u_1,\dots,u_k,u_1)$ of $\Cyk$ that contains at least one medium or high vertex,
has at least one sequence of paths that covers it. 
To verify this, note that 
$\Cy$ can always be covered in the following manner. If there are no high vertices within the copy, a single path of length $k$ can be constructed, starting from any of its medium vertices. Otherwise ($\Cy$ contains at least one high vertex),  the initial path begins at some high vertex, with each subsequent high vertex initiating a new path (which may be of length one if multiple high vertices appear consecutively).
As an example, for the cycle on the left hand side of Figure~\ref{fig:fits}, we would get the paths $(u_1)$, $(u_2,u_3,u_4,u_5,u_6)$, $(u_7,u_8)$, and for the cycle on the right, 
$(u_1,u_2)$, $(u_3,u_4,u_5,u_6)$, $(u_7,u_8)$.

\begin{definition}[fitting]\label{def:fitting}
 Let $\Cy = (v_1,\dots,v_k,v_1)$ be a copy of $\Cyk$ in $G$ and let $\ogamma$ and $\om$ be positive integers such that $\ogamma \leq \sqrt{\om}$.
 For a sequence of positive integers $\vec{x} = \langle x_1,\dots, x_r\rangle$ such that $\sum_{q=1}^r x_q = k$, we say that $\vec{x}$ \textsf{fits} $\Cy$ if there exist a sequence of paths $\pi_1,\dots,\pi_r$ that covers $\Q$ such that $|\pi_q|=x_q$ for every $q\in [r]$. We say in such a case that $\pi_1,\dots,\pi_r$ \textsf{corresponds} to $\vec{x}$.

    We denote by $\nu_{\vec{x}}(\Cy)$ the number of different sequences 
    of paths
    covering $\Cy$ that correspond to $\vec{x}$ and by $\nu(\Cy)$, the sum over all $\vec{x}$ that fit $\Cy$\tnew{,} of $\nu_{\vec{x}}(\Cy)$. 
 \end{definition}

\ifnum\stoc=0
\begin{figure}
\begin{center}
\mbox{
\includegraphics[width=0.75\linewidth]{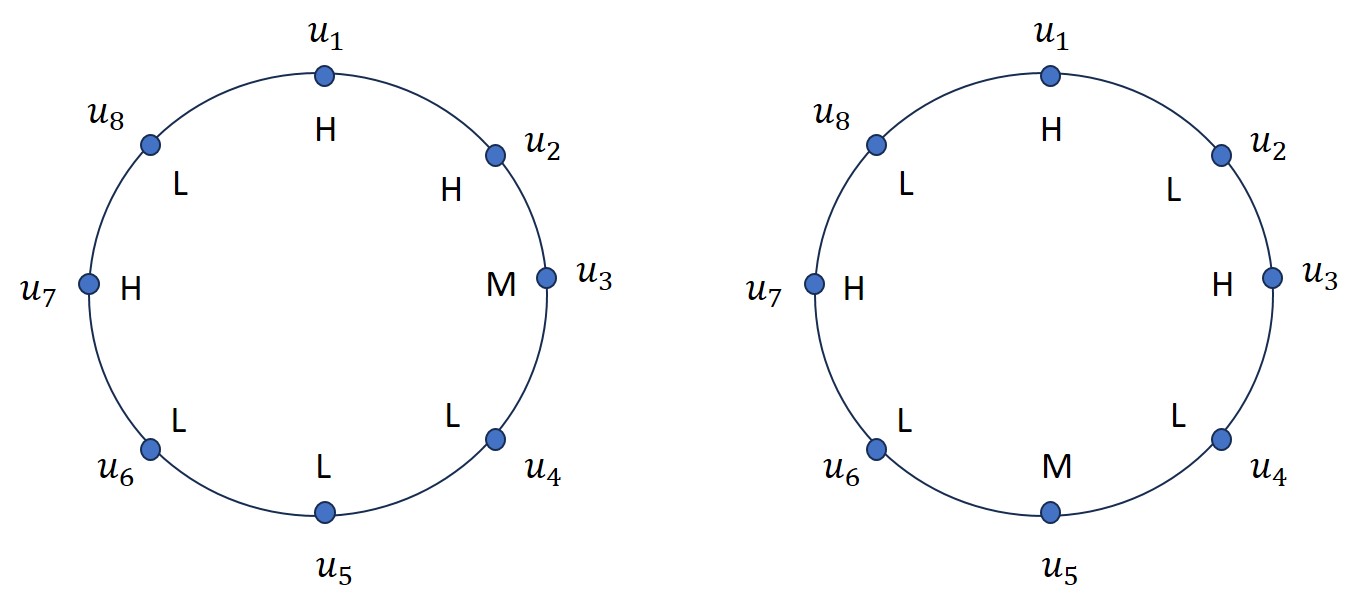}}
    \caption{\small The letters `H', `M' and `L' signify wether a vertex is high, medium or low, respectively.
    The sequence $<2, 5, 1>$ fits the cycle on the left. To verify this, consider the following corresponding sequence of paths that covers the cycle: $\pi_1 = (u_1,u_8)$, $\pi_2 = (u_7,u_6,u_5,u_3,u_2)$, $\pi_3 = (u_1)$.
    The sequence $<1, 1, 4, 2>$ also fits this cycle. To verify this, consider the following sequence of paths that covers the cycle:
    $\pi_1 = (u_1)$, $\pi_2 = (u_2)$, $\pi_3 = (u_3,u_4,u_5,u_6)$, $\pi_4 = (u_7,u_8)$.
The sequence $<2, 2, 2, 2>$ fits the cycle on the right. There are several sequences of paths that cover the cycle and correspond to this sequence of lengths. For example the sequence $\pi_1 = (u_1, u_2)$, $\pi_2 = (u_3, u_4)$, $\pi_3 = (u_5, u_6)$, $\pi_4 = (7,8)$, and the sequence $\pi'_1 = (u_3, u_2)$, $\pi'_2 = (u_1, u_8)$, $\pi'_3 = (u_7, u_6)$, $\pi'_4 = (5,4)$.
   }
    \label{fig:fits}
    \end{center}
\end{figure}
\else
\begin{figure}
 \begin{center}
\includegraphics[width=0.9\linewidth]{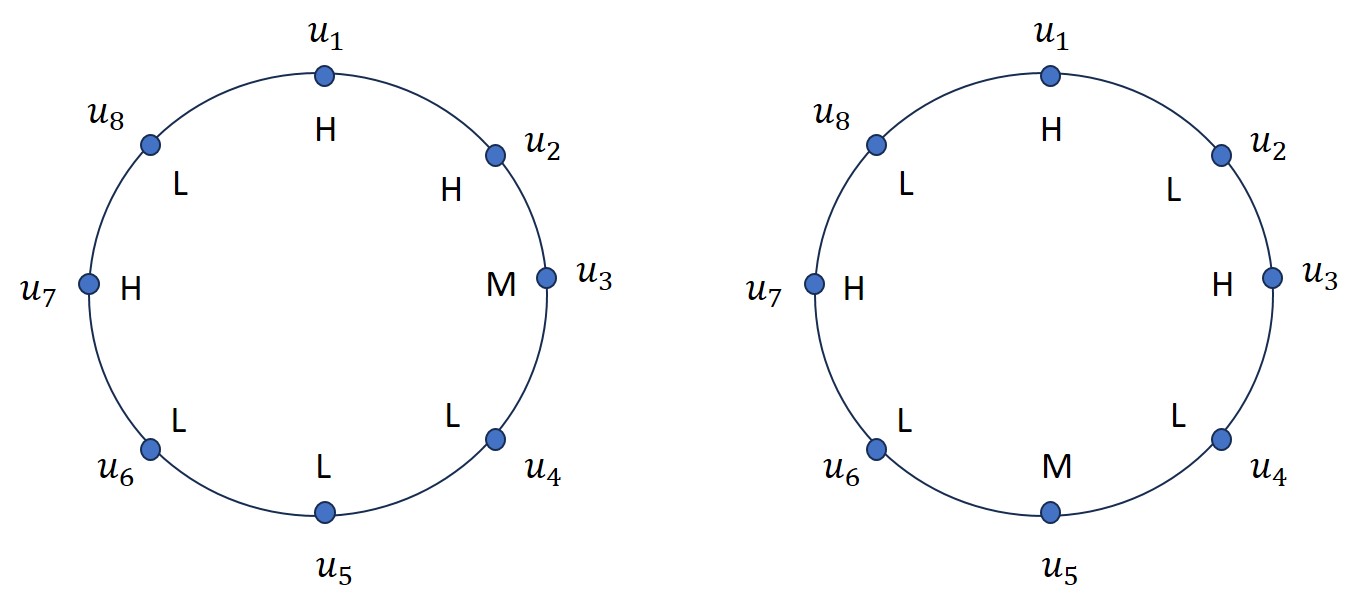}
    \Description{}
    \caption{\small The letters `H', `M' and `L' signify wether 
    a vertex is high, medium or low, respectively.
    The sequence $<2, 5, 1>$ fits the cycle on the left. To verify this, consider the following corresponding sequence of paths that covers the cycle: $\pi_1 = (u_1,u_8)$, $\pi_2 = (u_7,u_6,u_5,u_3,u_2)$, $\pi_3 = (u_1)$.
    The sequence $<1, 1, 4, 2>$ also fits this cycle. To verify this, consider the following sequence of paths that covers the cycle:
    $\pi_1 = (u_1)$, $\pi_2 = (u_2)$, $\pi_3 = (u_3,u_4,u_5,u_6)$, $\pi_4 = (u_7,u_8)$.
The sequence $<2, 2, 2, 2>$ fits the cycle on the right. There are several sequences of paths that cover the cycle and correspond to this sequence of lengths. For example the sequence $\pi_1 = (u_1, u_2)$, $\pi_2 = (u_3, u_4)$, $\pi_3 = (u_5, u_6)$, $\pi_4 = (7,8)$, and the sequence $\pi'_1 = (u_3, u_2)$, $\pi'_2 = (u_1, u_8)$, $\pi'_3 = (u_7, u_6)$, $\pi'_4 = (5,4)$.
   }
    \label{fig:fits}
    \end{center}
\end{figure}
\fi

We extend the above definition to a general motif $\F$ that includes at least one Hamiltonian cycle.

\begin{definition}[number of fittings]\label{def:number-fittings}
For a copy $\f$ of $\F$ in $G$, we define 
$\nu(\f) = \sum_{\Cy \in \ham(\f)} \nu(\Cy)$.
\end{definition}

\alg{alg:sampHMCopies}{
	{\bf \sampHMCopies}$\;(\F,\ogamma,\om,\calD)$ 
	\begin{enumerate}[itemsep=.2em, leftmargin=*]
		\item \label{step:select-vecx} Select, u.a.r., a sequence $\vec{x} = \langle x_1,\dots,x_r\rangle$ of positive integers such that $\sum_{q=1}^r x_q =k$. 
        \item For ($q=1$ to $r$) do:\label{step:select-path}
           \begin{enumerate}[itemsep=.2em]
              \item \label{step:select-hm} Invoke \sampHighVertex$(\calD,\ogamma)$. 
              If no vertex is returned, then
              \textbf{return} \emph{fail}.
            Otherwise, denote the returned vertex   by $v_{q,1}$. 
            
              \item \label{step:x-q-1} If $x_q = 1$:
                  \begin{itemize}
                      \item  If $d(v_{q,1}) \leq \sqrt{\om}$, then \textbf{return} \emph{fail}, 
                    otherwise \textbf{return} \emph{fail} with probability $1-\frac{\sqrt{\om}}{d(v_{q,1})}$. \newcomment{If the path contains a single vertex, then it must have degree greater than $\sqrt{\om}$. 
                    }
                  \end{itemize}  
              \item \label{step:first-edge} If $x_q > 1$, then choose an index $j \in [d(v_{q,1})]$ and query for $v_{q,1}$'s $j\th$ neighbor. Let this vertex be denoted $v_{q,2}$. If $d(v_{q,2}) > \sqrt{\om}$, then
              \textbf{return} \emph{fail}, otherwise,  \textbf{return} \emph{fail} with probability $1/2$.
              \newcomment{Sample a uniform neighbor of $v_{q,1}$ and verify it does not have high degree.}
             \item \label{step-other-edges} For ($y = 2$ to $x_q-1$) do
              \begin{itemize}
                \item Choose an index $j\in [\sqrt{\om}]$  u.a.r. 
                If $j>d(v_{q,y})$, then \textbf{return} \emph{fail}. 
                Otherwise query for $v_{q,y}$'s $j\th$ neighbor and let $v_{q,y+1}$ denote the neighbor returned.  If $d(v_{q,y+1}) > \sqrt{\om}$, then
              \textbf{return} \emph{fail}, otherwise,  \textbf{return} \emph{fail} with probability $1/2$. 
              \newcomment{Sample a uniform neighbor of $v_{q,1}$, such that each is sampled w.p. $\frac{1}{2\sqrt m}$. If the sampled neighbor is high, then fail.}
              \label{step:hm-sample-neighbor-light}
               \end{itemize}
          \end{enumerate}
      \item \label{step:check-cycle} If $v_{q,y} = v_{q',y'}$ for any $(q,y) \neq (q',y')$
      or $(v_{q,x_q},v_{q+1,1}) \notin E$ for some $1\leq q \leq r-1$, or
            $(v_{r,x_r}, v_{1,1})\notin E$, then \textbf{return} \emph{fail}. \newcomment{Verify that no vertex is sampled twice and that the sampled path cover results in a cycle.}
      \item \label{step:select-hm-copy} Let $\Q$ be the induced subgraph over the vertices $v_{1,1},\dots,v_{r,x_r}$. 
     If $\Q$ contains at least one copy of $\F$, then either select one such copy $\f$ or \textbf{return} \emph{fail}, where for each copy, the probability that it is selected is $1/\kappaF$.  \newcomment{Recall: $\kappaF$ is defined in Definition~\ref{def:nF}}
      \item \label{step:return-hm-copy} Return $\f$ with probability $\frac{1}{\nu(\f)}$, and with probability $1-\frac{1}{\nu(\f)}$ \textbf{return} \emph{fail}. \newcomment{Recall:  $\nu(f)$ is defined in Definition~\ref{def:number-fittings}.}
  \end{enumerate}
     
}

\begin{lemma}\label{lem:samp-high-med}
\sloppy
    Algorithm \sampHMCopies$(\F,\ogamma,\om,\calD)$ either returns a copy of $\F$ in $G$ that contains some vertex with degree exceeding 
    $\ogamma$ or fails. If $\calD$ is $(\oeps,\ogamma,\om)$-degrees-typical, then each such copy is returned  with probability in $(1\pm \oeps)^k\cdot \frac{ 1}{ \om^{k/2} \cdot 2^k\cdot \kappaF }$. 
    The query complexity of the algorithm is $O(k^2)$ and the running time is at most $O(k!\cdot k^2)$.
\end{lemma}
As noted following Lemma~\ref{lem:samp-low}, both the query complexity and the running time of the algorithm may be smaller than what is stated in the lemma, depending on $\F$.
For example if $\F = \Cyk$, then the running time and query complexity are both $O(k)$.

\ifnum\stoc=0
\begin{proof}
The algorithm performs at most $k$ degree queries (one on each of the vertices $v_{q,y}$ for $q\in [r]$, $y\in [x_r]$), at most $k-1$ neighbor queries (one on each of the vertices $v_{q,y}$ for $q\in [r]$, $y\in [x_q-1]$), and less than ${k \choose 2}$ pair queries (to check whether $\{v_{x_r},v_{1,1}\}\in E$ and to obtain the induced subgraph $\Q$). The upper bound on the query complexity of the algorithm follows. The upper bound on the running time is established as in the proof of Lemma~\ref{lem:samp-low}. 

Given $\ogamma$ and $\om$, in what follows we shall say that a vertex $v$ is \emph{low} if $d(v) \leq \ogamma$, that it is \emph{high} if $d(v) > \sqrt{\om}$ and otherwise it is \emph{medium}. 

The fact that the algorithm either returns a mixed copy or fails, follows directly from the description of the algorithm. We hence turn to analyzing the probability that each such copy is returned, where from here on we condition on the data structure $\calD$ being $(\oeps,\ogamma,\om)$-degrees-typical.  This implies that for every vertex $u$, if $d(u)>\ogamma$, then the probability that $u$ is returned by 
\sampHighVertex$(\calD,\ogamma)$, which we denote by
$\phi_{\calD}(u)$, is in $(1\pm \oeps)\cdot \frac{d(u)}{2\om}$,
and if $d(u)\leq \ogamma$, then this probability is $0$.

    Consider any fixed mixed copy $\f^*$ of $\F$ in $G$, and
    let $\Cy^*$ be one of the $\hF$ copies of $\Cyk$ in $\ham(\f^*)$ (Hamiltonian cycles of $\f^*$).
    Using the notions introduced in Definitions~\ref{def:cover} and~\ref{def:fitting}, let $\vec{x}^* = \langle x^*_1,\dots,x^*_r\rangle$, $\sum_{q=1}^r x^*_q = k$ be a sequence that fits $\Cy^*$ and consider a fixed sequence of corresponding paths $\big\langle \pi^*_1 = (u_{1,1},\dots,u_{1,x^*_1}), \dots, \pi^*_r = (u_{r,1},\dots,u_{r,x^*_r})\big\rangle $ that cover $\Cy^*$.  Recall that (in particular) this means that for each $q\in [r]$ we have the following:
    (1) $|\pi^*_q| = x^*_q$; (2) $d(u_{q,1}) > \ogamma$ and if $x^*_q =1$, then $d(u_{q,1}) > \sqrt{\om}$ ; (3) $d(u_{q,y}) \leq \sqrt{\om}$ for every $2\leq y \leq x^*_q$.

    Suppose $\vec{x}^*$ is selected in Step~\ref{step:select-vecx} of \sampHMCopies. 
    The probability that $v_{q,y} = u_{q,y}$ for every $1 \leq q \leq r$ and $1 \leq y < x^*_q$
    is the product of $p(\pi^*_1),\dots,p(\pi^*_r)$, where $p(\pi^*_q)$ is the probability of hitting the path $\pi^*_q$ in the $q$-th iteration of Step~\ref{step:select-path}, and is computed precisely next. 
    \begin{enumerate}
        \item \sloppy If $x^*_q=1$, then by Steps~\ref{step:select-hm} and~\ref{step:x-q-1},
        \begin{equation} 
        p(\pi^*_q) = \phi_{\calD}(u_{q,1}) \cdot \frac{\sqrt{\om}}{d(v_{q,1})} \in (1\pm \oeps)\cdot\frac{ d(v_{q,1})}{2{\om}} \cdot \frac{\sqrt{\om}}{d(v_{q,1})}  = (1\pm \oeps) \cdot \frac{1}{2\sqrt{\om}}\;,
        \end{equation}
                  where we used the fact that $\calD$ is $(\oeps,\ogamma,\om)$-degrees-typical.
        \item Otherwise ($x^*_q \geq 2$), by Steps~\ref{step:select-hm}, ~\ref{step:first-edge} and~\ref{step-other-edges},
          \begin{eqnarray} 
          p(\pi^*_q) &=& \phi_{\calD}(u_{q,1}) \cdot \frac{1}{d(v_{q,1})}\cdot \frac{1}{2}\cdot \left(\frac{1}{\sqrt{\om}}\cdot\frac{1}{2}\right)^{x^*_q-2} \\
          &\in& (1\pm \oeps)\cdot\frac{ d(v_{q,1})}{2{\om}} \cdot \frac{1}{2d(v_{q,1})}\cdot \left(\frac{1}{2\sqrt{\om}}\right)^{x^*_q-2} \\
          &=&
          (1\pm \oeps)\cdot \left(\frac{1}{2\sqrt{\om}}\right)^{x^*_q}\;,
          \end{eqnarray}
     where again we used the fact that $\calD$ is $(\oeps,\ogamma,\om)$-degrees-typical.
    \end{enumerate}
    Since $\sum_q x^*_q = k$,
we get that the overall probability that $v_{q,y} = u_{q,y}$ for every $1 \leq q \leq r$ and $1 \leq y < x^*_q$ is
\begin{equation}
    \prod_{q=1}^r p(\pi^*_q) \;\in\; (1\pm \oeps)^k \cdot \left(\frac{1}{2\sqrt{\om}}\right)^k\;.
\end{equation}
Conditioned on $v_{q,y} = u_{q,y}$ for every $1 \leq q \leq r$ and $1 \leq y < x^*_q$, the algorithm selects $\f^*$ (among the different copies of $\F$ in the induced subgraph $\Q$) in Step~\ref{step:select-hm-copy}
with probability $1/\kappaF$. Conditioned on $\f^*$ being selected, it is returned in Step~\ref{step:return-hm-copy} with probability $1/\nu(\f^*)$, where recall that by Definition~\ref{def:number-fittings},
$$\nu(\f^*) = \sum_{\Cy\in \ham(\f^*)}\;\sum_{\vec{x} \mbox{ \footnotesize fits } \Cy} \nu_{\vec{x}}(\Cy)\;,$$ and $\nu_{\vec{x}}(\Cy)$ is the number of different sequences of paths that cover $\Cy$ and correspond to $\vec{x}$ (see Definition~\ref{def:fitting}). 

Let $\Pi_{\vec{x}}(\Cy)$ denote the set of sequences of paths that cover $\Cy$ and correspond to $\vec{x}$ (so that $\nu_{\vec{x}}(\Cy) = |\Pi_{\vec{x}}(\Cy)|$).
Then the overall probability that $\f^*$ is returned is
$$\sum_{\Cy\in \ham(\f^*)}\;\sum_{\vec{x} \mbox{ \footnotesize fits } C} \;\sum_{\langle\pi_1,\dots,\pi_{r=|\vec{x}|}\rangle\in \Pi_{\vec{x}}(\Cy)}\; \prod_{q=1}^r p(\pi^*_q) \cdot \frac{1}{\kappaF} \cdot \frac{1}{\nu(\f^*)} \;\;\in\;\;  \frac{(1\pm \oeps)^k }{2^k \cdot \om^{k/2}\cdot \kappaF }\;,
$$
as claimed.
\end{proof}
\fi 

\subsection{Sampling copies}\label{subsec:sample-copy}
We next combine our two attempted sampling procedures to obtain a single attempted sampling procedure. 

\smallskip

\alg{alg:sample-copy}{
	{\bf \SampleWithEstimate}$\;(n, \F,\ogamma, \om, \calD)$ 
	\begin{enumerate}[, leftmargin=*]
 \item Set $\Bl =  n\cdot \ogamma^{k-1} \cdot \kappaF$,\; $\Bmh =  2^k\cdot \om^{k/2} \cdot \kappaF$ \tnew{and $B=\Bl+\Bmh$}.
 
       \item 
       With probability $\frac{\Bl}{B}$ invoke \sampLowCopies$(\F,\ogamma)$ and with probability $\frac{\Bmh}{B}$ invoke \sampHMCopies$(\F,\ogamma,\om,\calD)$.
       \item If a copy of $\F$ was returned, then return this copy. Otherwise \textbf{return} \emph{fail}. 
       \end{enumerate}
}

\begin{lemma}\label{lem:sampleA}
    If \SampleWithEstimate\ is invoked with  an $(\oeps,\ogamma,\om)$-degrees-typical data structure $\calD$  (see Definition~\ref{def:good-DS}) then it returns each copy of $\F$ with probability in $(1\pm\oeps)^k\cdot \frac{1}{B}$.     The query complexity of the algorithm is $O(k^2)$ and the running time is  $O(k!\cdot k^2)$.
\end{lemma}

\ifnum\stoc=0
\begin{proof}
 By Lemma~\ref{lem:samp-low}, an invocation of \sampLowCopies$(\F,\ogamma)$  
 returns each specific low copy of $\F$
  with probability exactly $\frac{1}{\Bl}$. 
By Lemma~\ref{lem:samp-high-med}, If $\calD$ is an $(\oeps,\ogamma,\om)$-degrees-typical data structure, then an invocation of \sampHMCopies$(\F,\ogamma,\om,\calD)$ returns each specific mixed copy of $\F$ with probability in $(1\pm\oeps)^k\cdot \frac{1}{ \Bmh}$. 
The lemma follows by the description of the algorithm and Lemmas~\ref{lem:samp-low} and~\ref{lem:samp-high-med}.
\ForFuture{Explicitly compute the probability}
\end{proof}
\fi 

\subsection{The approximate counting algorithm}\label{subsec:approx-count}

We now present an algorithm that, given as ``advice'' an estimate $\onF$ of $\countF$ and an estimate $\om$ of $m$ (which are not necessarily good estimates), returns an approximate value of $\hnF$, which satisfies the following conditioned on $\om \geq m$. If the estimate $\onF$ is upper bounded by the true value, i.e., $\onF\leq \countF$, then $\hnF\in(1\pm\eps)$ with high probability, and if 
$\onF > \countF$, then $\hnF$ is not too high with some probability that is not too small. 

\smallskip

\newcommand{\settApproxAdvice}{\frac{B}{\onF} \cdot \frac{3\ln(4/\delta)}{(1-\oeps)^k\cdot\oeps^2}}

\alg{alg:approx-with-estimate}{
	{\bf \ApproxWithEstimate}$\;(n,\F,\eps,\delta,\onF,\om)$ 
	\begin{enumerate}[itemsep=.2em, leftmargin=*]
		\item Set $\oeps = \eps/(6k)$,\;  $\ogamma = \onF^{1/k}$,\;  
 $B = n\cdot \ogamma^{k-1}\cdot \kappaF + 2^k\cdot \om^{k/2} \cdot \kappaF$.
 \label{step:define-Bs}
     \item \label{step:preprocess} Invoke \preprocess$(n,\oeps,\min\{\delta/2,\eps/12\},\ogamma,\om)$ and let $\calD$ denote the returned data structure.
       \item For $i=1$ to $t =\settApproxAdvice $ do:\label{step:for-loop}
       \begin{enumerate}
       \item 
       Invoke \SampleWithEstimate$(n,\F,\ogamma, \om, \calD)$.
       \item If a copy of $\F$ was returned, then set $\chi_i = 1$ otherwise set $\chi_i = 0$.
           \end{enumerate}
       \item Set $\chi = \sum_{i=1}^t \chi_i$\;.
   \item Return $\hnF=\frac{B}{t}\cdot \chi$\;. 
	\end{enumerate}
}

\begin{theorem}\label{thm:approx-estimate}
Let $G$ be a graph over $n$ vertices and $m$ ordered edges, and let $\eps,\delta\in(0,\frac{1}{2})$. Further, let $\F$ be a Hamiltonian motif over $k$ vertices,
$\onF$ an estimate of $\countF$, and $\om$ an estimate of $m$. Given 
 query access to  $G$,
the algorithm \ApproxWithEstimate$(n,\F,\eps,\delta,\onF,\om)$ 
 satisfies the following:
\begin{enumerate}
   \item  If $\om \geq m$ and $\onF \leq \countF$, then with probability at least $1-\delta$, \ApproxWithEstimate{} returns a value $\hnF$ such that $\hnF\in(1\pm\eps)\cdot \countF$.
\item  If $\om \geq m$ and $\onF>\countF$, then  \ApproxWithEstimate{} returns a value $\hnF$, such that with probability at least $\eps/4$, $\hnF \leq (1+\eps)\countF$.
\item \sloppy The 
query complexity of \ApproxWithEstimate{} is $O\left(\frac{n}{\onF^{1/k}}+\frac{\om^{k/2}}{\onF}\right)\cdot \frac{k^4 \cdot \kappaF\cdot \log^2(n/\delta)}{\eps^2}$. 
\item \sloppy The 
running time of \ApproxWithEstimate{} is  $O\left(\frac{n}{\onF^{1/k}}+\frac{\om^{k/2}}{\onF}\right)\cdot \frac{k!\cdot k^4 \cdot \kappaF \cdot \log^2(n/\delta)}{\eps^2}$.
\end{enumerate}
\end{theorem}

\ifnum\stoc=0
 \begin{proof}
By Lemma~\ref{lem:preproc}, with probability at least $1-\min\{\delta/2,\eps/12\}$, the data structure $\calD$ returned by
\preprocess\ in Step~\ref{step:preprocess} is an $(\oeps,\ogamma,\om)$-degrees-typical data structure (recall Definition~\ref{def:good-DS}). Let $E_{dt}$ denote this event (where $dt$ stands for ``degrees-typical'') and condition on it. 

Observe that the setting of $B$ in Algorithm \ApproxWithEstimate\ is the same as the setting of $B$ in the Algorithm \SampleWithEstimate.
By Lemma~\ref{lem:sampleA}, (conditioned on $E_{dt}$), each copy of $\F$ is returned with probability in 
$(1\pm\oeps)^k\cdot  \frac{1}{B}$.
Therefore, for every iteration $i$, $\Ex[\chi_i]=\Pr[\chi_i=1] \in (1\pm\oeps)^k\cdot \frac{\countF}{B}$. 
%
It follows that $\Ex[\chi]\in (1\pm\oeps)^k\cdot t\cdot \frac{\countF}{B}$.

We first consider the case that $\onF>\countF$.
By the above, 
$\Ex[\chi]\leq  (1+\oeps)^k \cdot t\cdot \frac{\countF}{B}$.
Since $\oeps=\eps/(6k)$ we have $(1+\oeps)^{k}\leq (1+2k\cdot\oeps)=1+\eps/3$. Thus, 
$\EX[\chi]\leq (1+\eps/3)\cdot t\cdot \frac{\countF}{B}$. 
 By Markov's inequality, for $\eps<1$,
\begin{align}\label{eq:markov}
\Pr\left[\chi>(1+\eps/2)\cdot \EX[\chi]\right]<\frac{1}{1+\eps/2}<1-\eps/3\;.
\end{align}
Denote by $E_{ub}$ the event that $\chi\leq (1+\eps/2)\EX[\chi]$ (where $ub$ stands for ``upper-bounded'').
By Lemma~\ref{lem:preproc} and by Equation~\eqref{eq:markov}, $\Pr[\overline{E}_{dt}]+\Pr[\overline{E}_{ub}\mid E_{dt}]< \eps/12+1-\eps/3=\tnew{1-\eps/4}$.
 Therefore, with probability at least $\eps/4$, 
$\chi<(1+\eps/2)(1+\eps/3)\cdot t\cdot \frac{\countF}{B}\leq (1+\eps) \cdot t\cdot \frac{\countF}{B}.$ 
By the setting of $\hnF=\frac{B}{t}\cdot \chi$, 
with probability at least 
 \tnew{$\eps/4$},
$\hnF<(1+\eps)\countF$, and the second item in the theorem holds.

We next consider the case that $\onF\leq \countF$. 
\tnew{Recall that conditioned on the event $E_{dt}$, $\Ex[\chi]\in (1\pm\oeps)^k\cdot t\cdot \frac{\countF}{B}$.}
By the multiplicative Chernoff's bound and by the setting of $t=\settApproxAdvice$ in Step~\ref{step:for-loop}, for the case that $\onF\leq \countF$ \tnew{and conditioned on the event $E_{dt}$},
\begin{align}\label{eq:chernoff}
\Pr&\Big[\left|\chi-\Ex[\chi]\right|>\oeps\cdot \Ex[\chi]\Big]\leq 
2\exp\left(-\frac{\oeps^2\cdot \Ex[\chi]}{3}\right)
\\ & \leq 2\exp\left(-\frac{\oeps^2\cdot (1-\oeps)^k\cdot \frac{B}{\onF}\cdot \frac{3\ln(2/\delta)}{(1-\oeps)^k\cdot\oeps^2}\cdot \frac{\countF}{B}}{3}\right)\leq \frac{\delta}{2}\;.
\end{align}
Denote by $E_{ga}$ the event that $\chi\in(1\pm\oeps)\EX[\chi]$ (where $ga$ stands for ``good approximation''). 
Since $\hnF=\frac{B}{t}\cdot \chi$,  conditioned on $E_{dt}$ and $E_{ga}$,
$$\hnF\in (1\pm\oeps)^{k+1}\cdot \countF\in(1\pm\eps)\cdot \countF,$$ 
where 
we use the setting of $\oeps=\eps/(6k)$. 
By Lemma~\ref{lem:preproc} and by Equation~\eqref{eq:chernoff}, $\Pr[\overline{E}_{dt}] + \Pr[\overline{E}_{ga} \mid E_{dt}] \leq 
\min\{\delta/2,\eps/12\}+\delta/2\leq \delta$. 
Therefore, we get that with probability at least $1-\delta$, $\hnF\in(1\pm\eps)\cdot\countF$, 
  and  the first item in the theorem holds.

\sloppy
We turn to analyze the query 
complexity and running time of \ApproxWithEstimate. 
By Lemma~\ref{lem:preproc}, 
the query complexity and running time of \preprocess$(n,\oeps,\min\{\delta/2,\eps/12\},\ogamma,\om)$ are $O\left(\frac{n}{\ogamma}\cdot \frac{\log(n/\delta)\cdot\log(1/\delta)}{\oeps^2}\right)$, where $\ogamma = \onF^{1/k}$. 
By Lemmas~\ref{lem:samp-low} and~\ref{lem:samp-high-med}, both procedures \sampLowCopies{} and \sampHMCopies{} have query complexity $O(k^2)$ and running time $O(k!\cdot k^2)$ respectively, and are each invoked at most $t$ times.
By the setting of $t$ in Step~\ref{step:for-loop},
\begin{align}
t &=O\left(\frac{B}{\onF}\cdot \frac{\log(1/\delta)}{(1-\oeps)^k\cdot \oeps^2}\right)=O\left(\frac{n\cdot \ogamma^{k-1}\cdot \kappaF + \om^{k/2}\cdot 2^k\cdot  \kappaF}{\onF}\cdot \frac{\log(1/\delta)}{\oeps^2} \right)
\\&=O\left(\frac{n}{\onF^{1/k}}+\frac{\om^{k/2}}{\onF} \right)\cdot \frac{2^k\cdot k^2\cdot  \ln(1/\delta)\cdot \kappaF}
{\eps^2}\;.
\end{align}
Therefore, the third and fourth items in the theorem hold as well and the proof is completed.
 \end{proof}
 \fi

In order to eliminate the need for 
advices $\om$ and $\onF$, we use the following two results.
\ifnum\stoc=1
The first result  is ~\cite[Theorem 1.1]{GR08} for estimating the number of edges in a graph with expected complexity $O\left(\frac{n\log^2 n}{\sqrt m}\right)$.
\else
The first result  is  for estimating the number of edges in a graph.
\begin{theorem}[Theorem 1.1 of \cite{GR08}, special case, restated.] \label{thm:GR}
There exists an algorithm that,  given  query access to a graph $G$ over $n$ vertices and  $m$ ordered edges, 
returns a value that with probability at least $2/3$, is in $[m, 2m]$.
The expected query complexity and running time of the algorithm are $O\left(\frac{n\log^2 n}{\sqrt m}\right)$.
\end{theorem}
\fi


 \ifnum\stoc=1
The second result  is a general  geometric search scheme, which we refer to as $\search{}$, presented in~\cite[Theorem 18]{cliquesSicomp}. 
The algorithm $\search{}$
   performs a geometric search for $\onF$, starting from its maximal possible value, and decreasing it in each step. 
Roughly, for each ``guessed'' value of $\onF$, it invokes
 \ApproxWithEstimate{}\ in order to decide whether to half the guess or halt and return an estimate.
With sufficiently high probability, the search is ensured to stop only after $\onF$ 
is smaller (but not much smaller) than $\countF$, resulting in a  $(1\pm \eps)$-estimate of $\countF$.
 \else
  Recall that \ApproxWithEstimate\ receives as input a parameter $\onF$, and intuitively works as if $\onF \approx \countF$. More precisely, by Theorem~\ref{thm:approx-estimate}, if $\onF \leq \countF$ (in addition to $\om \geq m$), then with probability at least $1-\delta$, the algorithm returns a $(1\pm \eps)$-estimate of $\countF$, while if $\onF > \countF$, then with non-negligible probability, it returns an estimate that is not much larger than $\countF$. The second result, stated below, essentially implies that under such conditions (and an additional one that \ApproxWithEstimate\ satisfies), the following holds. It is possible to perform a search on $\onF$ using invocations of \ApproxWithEstimate, starting from its maximal possible value, and decreasing it in each step.
With sufficiently high probability, the search is ensured to stop only after $\onF$ goes below $\countF$, resulting in a  $(1\pm \eps)$-estimate of $\countF$.
 \fi

\ifnum\stoc=0
\begin{theorem}[Theorem 18 in~\cite{cliquesSicomp}] \label{thm:search}
	Let  $\invokeA\left(\ov,\eps,\delta, \vecV\right)$ be an algorithm that is given parameters $\ov,\eps, \delta$ and possibly an additional set of parameters denoted $\overrightarrow V$,
	for which the following holds.
\begin{enumerate}
	\item\label{it:oa-good} If $\ov \in [\vvv/4,\vvv]$, then with probability at least $1-\delta$, $\mA$ returns a value $\hv$ such that
	$\hv\in(1\pm \eps)\vvv$.
	\item\label{it:oa-big} If $\ov > \vvv$, then $\mA$ returns a value $\hv$, such that with probability at least $\eps/4$, $\hv \leq (1+\eps)\vvv$.
	\item\label{it:exp-t-A} The expected running time of $\mA$, denoted $\ert{\mA\left(\ov,\eps,\delta, \vecV\right)}$,
 is monotonically non-increasing with $\ov$ and furthermore, if $\ov < \vvv$, then 
$\ert{\mA\left(\ov,\eps,\delta, \vecV\right)} \leq \ert{\mA\left(\vvv,\eps,\delta, \vecV\right)}\cdot (\vvv/\ov)^\ell$
for some constant $\ell > 0$.
	\item\label{it:max-t-A} The maximal running time of $\mA$ is $M$. 
\end{enumerate}
Then there exists an algorithm \search{} that, when given an upper bound $U$ on $\vvv$, a parameter $\eps$, a set of parameters $\vecV$ and access to $\algA$, returns a value $X$ such that the following holds. 
\begin{enumerate}
	\item\label{it:Ap-good}  \search$(\algA,U,\eps,\vecV)$  returns a value $X$ such that $X \in (1\pm\eps)\vvv$ with probability at least $4/5$.
	\item\label{it:exp-t-Ap} The expected running time of \search$\left(\algA,U,\eps,\vecV\right)$ is $\ert{\mA\left(\vvv,\eps,\delta, \vecV\right)} \cdot \poly( \log(U), 1/\eps, \ell)$ for $\delta =\Theta\left(\frac{\eps}{2^{\ell}(\ell + \log \log (U))}\right)\;.$
	\item\label{it:max-t-Ap} The maximal running time of \search\ is $M \cdot \poly( \log (U), 1/\eps, \ell)$.
\end{enumerate}
\end{theorem}

\fi

\smallskip

\alg{alg:approx-search}{
{\bf \Approx}$\;(n,\eps,  \F)$ 
\begin{enumerate}[itemsep=.2em, leftmargin=*]
\item Let $m'$ be the median value returned over $2\ln(n^{2k})$ invocations of the algorithm \ifnum\stoc=1
from~\cite{GR08}
\fi
for approximating the number of 
\ifnum\stoc=1
edges.
\else
edges, referred to 
in Theorem~\ref{thm:GR}. 
\fi
Let $\om=\min\{m',n^2\}$.
\item Let $U=\kappaF\cdot \om^{k/2}$.
\item 
\ifnum\stoc=0
Invoke $\search($\ApproxWithEstimate$,U,\eps,\vec V)$ with $\vec V=(n,\F,\om)$ 
where $\search$ is as described in Theorem~\ref{thm:search}.
\else
Invoke the algorithm $\search$  with access to 
\ApproxWithEstimate, parameters $n,F,\eps,\delta, \om$, and initial search value $U$.
\fi
\end{enumerate}
 (If the running time at any point goes above $n$, then simply query the entire graph and store it in memory.  Then  invoke the algorithm of~\cite{fichtenberger2020sampling} and answer every query with the stored information.)
}


\ifnum\stoc=0
\begin{proofof}{Theorem~\ref{thm:main}}
%
By Theorem~\ref{thm:approx-estimate}, Algorithm \ApproxWithEstimate{} meets the conditions of Theorem~\ref{thm:search} where for the third item, $\ell=1$, and for the fourth item, its 
running time is upper bounded by
$O\left(\frac{n}{\onF^{1/k}}+\frac{\om^{k/2}}{\onF}\right)\cdot \frac{k!\cdot k^4 \cdot \kappaF\cdot \log^2(n/\delta)}{\eps^2}$.

By the first item in Theorem~\ref{thm:search}, with probability at least $4/5$, \Approx{} returns a $(1\pm\eps)$ multiplicative approximation of $\countF$, and so correctness holds.

We turn to analyze the expected query complexity and running time. 
By Theorem~\ref{thm:GR}, the expected query complexity and running time of the first step are $O\left(\frac{n\cdot k\log^3 n}{\sqrt{m}}\right)$. Since for any $\F$, $\countF\leq m^{k/2}\cdot \kappaF$, this expression is upper bounded by $O\left(\frac{n\cdot \kappaF\cdot  k\log^3 n}{\countF^{1/k}}\right)$ and as will be clear momentarily, does not affect the overall asymptotic complexity of the algorithm.

By the second item in Theorem~\ref{thm:search}, the expected running time of \search$($\ApproxWithEstimate$,U,\eps,\vecV=\{n,\F,\om\})$ is $E_{rt}($\ApproxWithEstimate$(n,\F,\eps,\delta,\countF,\om))\cdot \poly(\log U, 1/\eps,\ell)$ for $\delta=\Theta\left(\frac{\eps}{2^{\ell}(\ell + \log \log (U))}\right)$.
By Theorem~\ref{thm:approx-estimate},  the 
query complexity and running time  of every invocation of \ApproxWithEstimate{} are $O(k^2)$ and $O(k!\cdot k^2)$  times $O\left(\frac{n}{\onF^{1/k}}+\frac{\om^{k/2}}{\onF}\right)\cdot \frac{2^k\cdot k^2\cdot \kappaF\cdot \log^2(n/\delta)}{\eps^2}$, respectively.
 Since we have $\ell=1$, it follows that $\delta=\Theta(\eps/\log(k\log n))$, and so we get that for $\om\in[m,2m]$ the expected query complexity and running time of the invocation of \search{} are 
\begin{align}\label{eq:rt}
O\left(\frac{n}{\countF^{1/k}}+\frac{m^{k/2}}{\countF}\right)
\cdot \poly(1/\eps, \log n,\exp(k))\;.
\end{align}

Now consider the event that $\om\notin[m,2m]$.
By Theorem~\ref{thm:GR}, the probability of this event is at most $1/n^{2k}.$ 
In that case, the maximal running time of \ApproxWithEstimate\ is when invoked with $\onF=1$ in which case it is $O\left(n+\om^{k/2} \right)\cdot \frac{2^k\cdot k^4 \cdot \ln(1/\delta)\cdot \kappaF}{\eps^2}\leq
n^k \cdot \frac{2^k\cdot k^4 \cdot \ln(1/\delta)\cdot \kappaF}{ \eps^2}$ (since  $\om\leq n^2$). We get that the contribution of the event $m\notin[m,2m]$ to the expected running time of \Approx{} is $\frac{1}{n^{2k}}\cdot n^k \cdot \frac{2^k\cdot k^4 \cdot \ln(1/\delta)\cdot \kappaF}{\eps^2}=\frac{2^k\cdot k^4 \cdot \ln(1/\delta)\cdot \kappaF}{n^k \cdot \eps^2}$ and therefore it also does not contribute to the  asymptotics of the query complexity and running time.
Therefore, the expected query complexity and running time are as stated in Equation~\eqref{eq:rt}. 

Finally, by the last line in \Approx, the algorithm never exceeds $O(n+m)$ queries, as once it exceeds $n$ queries,  we query the entire graph using additional $m$ queries (by going over all vertices and querying all of their neighbors). We then  answer any subsequent queries using the stored information. Also note that when this event occurs we can answer uniform edge queries, so that it is more beneficial to  continue with the algorithm of~\cite{fichtenberger2020sampling}
as it is simpler and also has better dependence on $k$. Thus we get that the maximum query complexity is $O(n+m)$ and that in that case the running time is $O^*\left(n+m+\frac{m^{k/2}}{\countF}\right)$.
%
%
\end{proofof}
\fi 

As noted in the introduction, it is possible to bound the query complexity and running time of \Approx\ with high probability and not just in expectation by slightly modifying the ``outer'' search algorithm and its analysis. All that is needed is to slightly decrease the probability that the search continues after it reaches $\onF \leq \countF/4$, and observe that the probability that it does not stop after any $r$ additional steps decreases exponentially with $r$.

\ifnum\stoc=0
\subsection{Sampling copies pointwise-close to uniform}\label{sec:uniform-sampling}
In this section we prove 
Theorem~\ref{thm:sampling}. \ForFuture{Restate the theorem here}

\begin{proofof}{Theorem~\ref{thm:sampling}}
The algorithm first 
obtains constant-factor estimates 
of $m$ and $\countF$. Specifically,
using the algorithm referred to in Theorem~\ref{thm:GR} it  can obtain an approximate value $\hm$  that is in $[m,2m]$ with probability at least $2/3$.
By invoking \Approx{} with $\eps=1/2$ it can obtain an approximate value $\hnF$ that is in $(1\pm1/2)\countF$ with probability at least $2/3$.
By repeating these invocations $\Theta(\log(1/\delta))$ times and taking the corresponding median values, it amplifies the success probability to $1-\delta/2$.  We henceforth condition on this event.

Using these good approximations of $m$ and $\countF$, and setting $\oeps = \eps/(6k)$ and $\hgamma = \min\{\hnF^{1/k},\sqrt{\hm}\}$,
the algorithm then  invokes \preprocess$(n,\oeps,\delta/2,\hgamma,\hm)$, which by Lemma~\ref{lem:preproc} returns a data structure $\calD$ that is $(\oeps,\hgamma,\hm)$-degrees typical with probability at least $1-\delta/2$.

This completes the preprocessing part of the algorithm, which by the above succeeds with probability at least $1-\delta$. The complexity of this part is dominated by the invocation of \Approx\ (which is as stated in Theorem~\ref{thm:main}).

Now, 
in order to obtain a copy of $\F$, the sampling algorithm invokes \SampleWithEstimate$(n,\F,\hgamma, \hm, \calD)$ repeatedly, until a copy is returned.
%
By Lemma~\ref{lem:sampleA}, conditioned on 
$\calD$ being $(\oeps,\hgamma,\hm)$-degrees-typical, each copy of $\F$ is returned by \SampleWithEstimate\ with probability in $(1\pm\oeps)^k\cdot \frac{1}{B}$ for $B =  (n\cdot \hgamma^k + 2^k\cdot \hm^{k/2})\cdot \kappaF$.
  Therefore, the maximum ratio between probabilities of obtaining two different copies over all pairs of copies is at most $(1+\oeps)^k/(1-\oeps)^k$.  Since $\oeps=\eps/(6k)$, this ratio is bounded by $(1+\eps/3)/(1-\eps/3) \leq 1 + \epsilon$.
  It follows that the ratio between the average probability over all copies and the probability to obtain any specific copy is at most $1 + \epsilon$ and at least $1/(1 + \epsilon)$, as required.  
 Namely, the distribution on returned copies is $\eps$-pointwise close to uniform. 
 The expected number of invocations until a copy is returned is $O\left(\frac{B}{\countF}\right)=O\left(\frac{n}{\countF^{1/k}}+\frac{2^k \cdot m^{k/2}}{\countF}\right)\cdot \kappaF$. 
\end{proofof}
\fi 

\section{Simplified attempted samplers for cliques and stars}\label{sec:simple-samplers-cliques-stars}
\ifnum\stoc=1
\balance
\fi
  In this section, we present attempted samplers for $k$-cliques and $k$-stars. When integrated into the approximate counting and nearly uniform sampling framework, these samplers yield algorithms that are simpler than those in previous studies.
  \ifnum\stoc=0
  We omit the proofs for the approximate counting and nearly uniform sampling results, as they closely resemble those 
  presented in Sections~\ref{subsec:approx-count} and~\ref{sec:uniform-sampling}.
  \fi

  The following definition will be useful in the description of the attempted samplers for both motifs.
\begin{definition}[an ordering on the vertices]\label{def:prec}
We say $u$ \textsf{precedes} $v$ and denote it by $u\prec v$,  if $d(u)< d(v)$  or if $d(u)=d(v)$ and $id(u)< id(v).$   
\end{definition}

\subsection{Simplified attempted samplers for cliques}\label{subsec:simple-samplers-cliques}


In the case where 
$\F$ is $k$-clique, we provide attempted sampling procedures that are even simpler than those described in Sections~\ref{subsec:samp-low-copy} and~\ref{subsec:samp-mixed-copy}.
In particular, we no longer need to use path covers and fitting sequences.
Let $\countK$ denote the number of $k$-cliques in the graph (i.e., as a shorthand for $n_{\K_k}$).

Here we partition the clique copies into three types -- low, medium and high -- 
according to the minimum-degree vertex in the clique (breaking ties  by ids).
That is, in the first type the minimum-degree vertex is of low degree $(\leq \ogamma)$, in the second ($\in[\ogamma+1, \sqrt{\om}]$), medium, and in the third it is of high degree ($>\sqrt{\om}$).

The procedure for sampling a low copy is almost the same as \sampLowCopies, which was presented in Section~\ref{subsec:samp-low-copy}, and the procedure \sampHMCopies{} is replaced by  two procedures: one for sampling medium clique copies and one for sampling high clique copies.
The latter two procedures are provided with a data structure $\calD$, which we henceforth assume is $(\oeps,\ogamma,\om)$-degrees typical. The query complexity and running time of all three procedures is $O(k^2)$.

\paragraph{Sampling low cliques.}
We first describe the modifications to \sampLowCopies\ that result in a attempted procedure for sampling a low clique.
In Step~\ref{step:sample-neighbor-light}, instead of querying for a neighbor of $v_i$, we always query for neighbors of $v_1$. That is, $v_1$ is chosen u.a.r., and $v_2, \ldots, v_k$ are neighbors of $v_1$ such that each is sampled with probability $1/\ogamma$. 
In addition, 
Step~\ref{step:select-low-copy} in \sampLowCopies{}
is replaced by a step which verifies that 
the subgraph induced by $v_1,\dots,v_k$ is a clique and that $v_1\prec\ldots\prec v_k$. The latter ensures that each clique can only be sampled in a single manner. Then in Step~\ref{step:return-low-copy}, the clique is returned with probability 1. 
Hence, 
the probability that any specific low clique is returned is exactly $\frac{1}{n\cdot \ogamma^{k-1}}$. 

\paragraph{Sampling medium cliques.}
The process starts by attempting to sample the first edge $(v_1, v_2)$
in the clique, according to the $\prec$ order.
This is done using \sampHighVertex$(\calD,\ogamma)$, which, if successful, returns a medium vertex, from which an incident edge is chosen uniformly (and otherwise the procedure  returns \emph{fail}). 
It then samples $k-2$ additional neighbors of $v_1$, denoted $v_3,\dots,v_k$,
each with probability $\frac{1}{\sqrt{\om}}$.
Here too, we verify that we obtained a clique and that $v_1\prec \ldots \prec v_k$. 
Since we assume that $\calD$ is $(\oeps,\ogamma,\om)$-degrees typical, by Lemma~\ref{lem:samp-high-med},
for each fixed medium clique $\f^*$, the probability that $\f^*$ is returned is in
$(1\pm \oeps)\cdot \frac{1}{2\om} \cdot \frac{1}{\sqrt{\om}^{k-1}} = (1\pm \oeps)\cdot \frac{1}{2\om^{k/2}}$.

\paragraph{Sampling high cliques.}
Recall that in high cliques, the min degree vertex has degree greater than $\sqrt{\om}$. The procedure for sampling a high clique attempts to sample such a  clique by  sampling $k$
high-degree vertices, each with probability approximately $\frac{1}{\sqrt{\om}}$.
More precisely,   \sampHighVertex$(\calD,\ogamma)$ is invoked $k$ times, and for each vertex $v_i$ returned, if $v_i$ is high, then $v_i$ is kept with probability $\frac{\sqrt{\om}}{d(v_i)}$.
The procedure  then verifies that the subgraph induced by $v_1,\dots,v_k$ is a clique and that $v_1\prec \dots\prec v_k$.
Since we assume that $\calD$ is $(\oeps,\ogamma,\om)$-degrees typical, by Lemma~\ref{lem:samp-high-med},
for each fixed high clique $\f^*$, the probability that 
$\f^*$ is returned is in
$(1\pm \oeps)^k\cdot \frac{1}{2^k\sqrt{\om}^k}$.


 
\subsection{Simplified attempted samplers for stars}
In this subsection we present simple samplers for $k$-stars, 
where a $k$-star is a central vertex connected to $k-1$ leaves.
As in the case of $k$-cliques, we crucially use the fact that stars have radius one. Let $\countS$ denote the number of $k$-stars in the graph (a shorthand for $n_{\S_k}$).

We note that our attempted star samplers provide somewhat weaker guarantees than all  previous described samplers: In the case that the estimate $\onS$ for $\nS$ given to the sampler is too small, some of the copies in the graph might not be returned at all. However, this does not affect the counting and final sampling algorithms, as they only require that the attempted samplers be well behaved when the estimate is not too small.

Let $\ogamma=\nS^{1/k}$.
We partition the star copies into two types -- low and non-low --  based on the degree of their central vertex and the threshold $\ogamma$, and present two samplers according to these types.  The  non-low star sampler is  provided with a data structure $\calD$, which we henceforth assume is $(\oeps,\ogamma,\om)$-degrees typical. The query complexity and running time of the two  procedures are $O(k)$.

\paragraph{Sampling low stars.}
Sampling a low copy of a star is almost identical to sampling a low clique. Namely, the procedure samples a vertex $v_1$ uniformly at random, and if $d(v_1)\leq \ogamma$, then it samples $k-1$ neighbors of $v_1$, each with probability $1/\ogamma.$
If all neighbor queries are successful, and $v_1\prec \ldots\prec v_{k}$,   the  copy is returned. Thus, each low copy is returned with probability exactly $\frac{1}{n\ogamma^{k-1}}$.

\paragraph{Sampling non-low stars.} To sample a non-low star copy, we first observe the following. If the number of stars in the graph is  upper bounded by $\onS,$ then the maximum degree in the graph cannot be higher than (roughly) $\onS^{1/(k-1)}$ as each vertex $v$ contributes $\binom{d(v)}{k-1}$ star copies to the total count.
Hence, we set $\odmax=\min\{\Theta(q, n\})$, where $q$ is the maximal integer such that $\binom{d(v)}{k-1}\leq \onS$ (so that $q=\Theta(\onS^{1/(k-1)})$\,). To sample a copy with a non-low central vertex, 
the procedure
invokes  \sampHighVertex$(\calD,\ogamma)$, and if  a non-low vertex $v_1$ is returned,  then it selects one of its incident edges uniformly at random.
It then samples $k-2$ additional neighbors of $v_1$, each with probability $1/\odmax$. If  all neighbor queries are successful, it verifies that $v_1\prec\ldots\prec v_k$, 
and returns the sampled copy.
Therefore, in the case that $\onS\geq \countS$ and conditioned on the typicality of $\calD$, every non-low star in the graph is sampled with probability in $(1\pm\oeps)\frac{1}{2\om\cdot \odmax^{k-2}}=\Theta\left(
\frac{1}{\min\big\{2\om\cdot n^{k-2},  2\om\cdot \onS^{1-1/(k-1)}\big\}}\right)$.

Based on these samplers, the query and time complexities of the approximate counting and nearly uniform sampling algorithms  are:
\ifnum\stoc=0
\begin{align}\label{eq:star-comp}
O^*\left(\frac{n}{\countS^{1/k}}+\min\left\{\frac{m\cdot n^{k-2}}{\countS},\frac{m}{\countS^{1/(k-1)}}\right\}\right)\;.
\end{align}
\else
\[
O^*\left(\frac{n}{\countS^{1/k}}+\min\left\{\frac{m\cdot n^{k-2}}{\countS},\frac{m}{\countS^{1/(k-1)}}\right\}\right)\;.
\]
\fi

This complexity matches that of~\cite{GRS11}, who also proved it to be optimal up to the dependencies in $1/\eps, \log n$ and $k.$ 

\ifnum\stoc=0
\bibliographystyle{alpha}
\else
{\bibliographystyle{ACM-Reference-Format}}
\fi
\bibliography{refs}

\newcommand{\etalchar}[1]{$^{#1}$}
\begin{thebibliography}{PPDSBLP16}

\bibitem[ABG{\etalchar{+}}18]{Aliak}
Maryam Aliakbarpour, Amartya~Shankha Biswas, Themis Gouleakis, John Peebles, Ronitt Rubinfeld, and Anak Yodpinyanee.
\newblock Sublinear-time algorithms for counting star subgraphs via edge sampling.
\newblock {\em Algorithmica}, 80(2):668--697, 2018.

\bibitem[AKK19]{AKK}
Sepehr Assadi, Michael Kapralov, and Sanjeev Khanna.
\newblock A simple sublinear-time algorithm for counting arbitrary subgraphs via edge sampling.
\newblock In {\em Proceedings of the 10th Innovations in Theoretical Computer Science Conference, (ITCS)}, pages 6:1--6:20, 2019.

\bibitem[Alo07]{Alon07}
Uri Alon.
\newblock Network motifs: theory and experimental approaches.
\newblock {\em Nature Reviews Genetics}, 8(6):450--461, 2007.

\bibitem[AN22]{assadi2022asymptotically}
Sepehr Assadi and Hoai-An Nguyen.
\newblock Asymptotically optimal bounds for estimating {H}-index in sublinear time with applications to subgraph counting.
\newblock In {\em Approximation, Randomization, and Combinatorial Optimization. Algorithms and Techniques, {APPROX/RANDOM}}, pages 48:1--48:20, 2022.

\bibitem[BBCG08]{becchetti2008efficient}
Luca Becchetti, Paolo Boldi, Carlos Castillo, and Aristides Gionis.
\newblock Efficient semi-streaming algorithms for local triangle counting in massive graphs.
\newblock In {\em Proceedings of the 14th ACM SIGKDD international conference on Knowledge discovery and data mining ({KDD})}, pages 16--24, 2008.

\bibitem[BBGM19]{BGM19a}
Anup Bhattacharya, Arijit Bishnu, Arijit Ghosh, and Gopinath Mishra.
\newblock Hyperedge estimation using polylogarithmic subset queries.
\newblock {\em ArXiv}, abs/1908.04196, 2019.

\bibitem[BBGM21]{bhattacharya2021triangle}
Anup Bhattacharya, Arijit Bishnu, Arijit Ghosh, and Gopinath Mishra.
\newblock On triangle estimation using tripartite independent set queries.
\newblock {\em Theory of Computing Systems}, 65(8):1165--1192, 2021.

\bibitem[BCL11]{bickel2011method}
Peter~J Bickel, Aiyou Chen, and Elizaveta Levina.
\newblock The method of moments and degree distributions for network models.
\newblock {\em The Annals of Statistics}, 39(5):2280, 2011.

\bibitem[BER21]{BER}
Amartya~Shankha Biswas, Talya Eden, and Ronitt Rubinfeld.
\newblock Towards a decomposition-optimal algorithm for counting and sampling arbitrary motifs in sublinear time.
\newblock In {\em Approximation, Randomization, and Combinatorial Optimization. Algorithms and Techniques, {APPROX/RANDOM}}, pages 55:1--55:19, 2021.

\bibitem[BFK01]{Faloutsos01}
Zhiqiang Bi, Christos Faloutsos, and Flip Korn.
\newblock The "dgx" distribution for mining massive, skewed data.
\newblock In {\em Proceedings of the Seventh ACM SIGKDD International Conference on Knowledge Discovery and Data Mining {(KDD)}}, page 17–26. Association for Computing Machinery, 2001.

\bibitem[BGL16]{benson2016higher}
Austin~R Benson, David~F Gleich, and Jure Leskovec.
\newblock Higher-order organization of complex networks.
\newblock {\em Science}, 353(6295):163--166, 2016.

\bibitem[BGM23]{Bishnu-triangles}
Arijit Bishnu, Arijit Ghosh, and Gopinath Mishra.
\newblock On the complexity of triangle counting using emptiness queries.
\newblock In {\em Approximation, Randomization, and Combinatorial Optimization. Algorithms and Techniques, {APPROX/RANDOM}}, pages 48:1--48:22, 2023.

\bibitem[BHPR{\etalchar{+}}20]{beame2020edge}
Paul Beame, Sariel Har-Peled, Sivaramakrishnan~Natarajan Ramamoorthy, Cyrus Rashtchian, and Makrand Sinha.
\newblock Edge estimation with independent set oracles.
\newblock {\em ACM Transactions on Algorithms (TALG)}, 16(4):1--27, 2020.

\bibitem[Bur04]{burt2004structural}
Ronald~S Burt.
\newblock Structural holes and good ideas.
\newblock {\em American journal of sociology}, 110(2):349--399, 2004.

\bibitem[CHEVW24]{virginia_cycles}
Keren Censor-Hillel, Tomer Even, and Virginia Vassilevska~Williams.
\newblock {Fast Approximate Counting of Cycles}.
\newblock In {\em Proceedings of the 51st International Colloquium on Automata, Languages, and Programming, ({ICALP})}, pages 37:1--37:20, 2024.

\bibitem[CJ19]{lp-survey}
Graham Cormode and Hossein Jowhari.
\newblock $l_p$ samplers and their applications: A survey.
\newblock {\em SIAM journal on}, February 2019.

\bibitem[CLW20]{Chen-is-edges}
Xi~Chen, Amit Levi, and Erik Waingarten.
\newblock Nearly optimal edge estimation with independent set queries.
\newblock In Shuchi Chawla, editor, {\em Proceedings of the 31st Annual {ACM-SIAM} Symposium on Discrete Algorithms, (SODA)}, pages 2916--2935, 2020.

\bibitem[DLM22]{DLM22}
Holger Dell, John Lapinskas, and Kitty Meeks.
\newblock Approximately counting and sampling small witnesses using a colorful decision oracle.
\newblock {\em SIAM Journal on Computing}, 51(4):849--899, 2022.

\bibitem[ELRS17]{ELRS}
Talya Eden, Amit Levi, Dana Ron, and C.~Seshadhri.
\newblock Approximately counting triangles in sublinear time.
\newblock {\em SIAM Journal on Computing}, 46(5):1603--1646, 2017.

\bibitem[EMR21]{EMR_multiple_edges}
Talya Eden, Saleet Mossel, and Ronitt Rubinfeld.
\newblock Sampling multiple edges efficiently.
\newblock In {\em Approximation, Randomization, and Combinatorial Optimization. Algorithms and Techniques, {APPROX/RANDOM}}, pages 51:1--51:15, 2021.

\bibitem[ENT23]{ENT_edge}
Talya Eden, Shyam Narayanan, and Jakub T{\v{e}}tek.
\newblock Sampling an edge in sublinear time exactly and optimally.
\newblock In {\em Symposium on Simplicity in Algorithms (SOSA)}, pages 253--260, 2023.

\bibitem[ER18a]{ER18_lbs}
Talya Eden and Will Rosenbaum.
\newblock Lower bounds for approximating graph parameters via communication complexity.
\newblock In {\em Approximation, Randomization, and Combinatorial Optimization. Algorithms and Techniques, {APPROX/RANDOM}}, pages 11:1--11:18, 2018.

\bibitem[ER18b]{ER17}
Talya Eden and Will Rosenbaum.
\newblock On sampling edges almost uniformly.
\newblock In {\em Symposium on Simplicity in Algorithms (SOSA)}, pages 7:1--7:9, 2018.

\bibitem[ERR19]{ERR19}
Talya Eden, Dana Ron, and Will Rosenbaum.
\newblock The arboricity captures the complexity of sampling edges.
\newblock In {\em Proceedings of the 46th International Colloquium on Automata, Languages, and Programming, ({ICALP})}, pages 52:1--52:14, 2019.

\bibitem[ERR22]{ERR_sampling_cliques}
Talya Eden, Dana Ron, and Will Rosenbaum.
\newblock Almost optimal bounds for sublinear-time sampling of $k$-cliques in bounded arboricity graphs.
\newblock In {\em Proceedings of the 49th International Colloquium on Automata, Languages, and Programming, ({ICALP})}, pages 56:1--56:19, 2022.

\bibitem[ERS19]{ERS19}
Talya Eden, Dana Ron, and C.~Seshadhri.
\newblock Sublinear time estimation of degree distribution moments: The arboricity connection.
\newblock {\em SIAM Journal on Discrete Math}, 33(4):2267--2285, 2019.

\bibitem[ERS20a]{ERS20}
Talya Eden, Dana Ron, and C.~Seshadhri.
\newblock Faster sublinear approximation of the number of \emph{k}-cliques in low-arboricity graphs.
\newblock In Shuchi Chawla, editor, {\em Proceedings of the 31st Annual {ACM-SIAM} Symposium on Discrete Algorithms, (SODA)}, pages 1467--1478, 2020.

\bibitem[ERS20b]{cliquesSicomp}
Talya Eden, Dana Ron, and C.~Seshadhri.
\newblock On approximating the number of $k$-cliques in sublinear time.
\newblock {\em SIAM Journal on Computing}, 49(4):747--771, 2020.

\bibitem[Fei06]{Feige-Avg}
Uriel Feige.
\newblock On sums of independent random variables with unbounded variance and estimating the average degree in a graph.
\newblock {\em SIAM Journal on Computing}, 35(4):964--984, 2006.

\bibitem[FG04]{SharpW1-hard}
J.~Flum and M.~Grohe.
\newblock The parameterized complexity of counting problems.
\newblock {\em SIAM Journal on Computing}, 33(4):538--547, 2004.

\bibitem[FGP20]{fichtenberger2020sampling}
Hendrik Fichtenberger, Mingze Gao, and Pan Peng.
\newblock Sampling arbitrary subgraphs exactly uniformly in sublinear time.
\newblock In {\em Proceedings of the 47th International Colloquium on Automata, Languages, and Programming, ({ICALP})}, pages 45:1--45:13, 2020.

\bibitem[GR08]{GR08}
Oded Goldreich and Dana Ron.
\newblock Approximating average parameters of graphs.
\newblock {\em Random Structures and Algorithms}, 32(4):473--493, 2008.

\bibitem[GRS11]{GRS11}
Mira Gonen, Dana Ron, and Yuval Shavitt.
\newblock Counting stars and other small subgraphs in sublinear-time.
\newblock {\em SIAM Journal on Discrete Math}, 25(3):1365--1411, 2011.

\bibitem[HBPS07]{hormozdiari2007not}
Fereydoun Hormozdiari, Petra Berenbrink, Nata{\v{s}}a Pr{\v{z}}ulj, and S~Cenk Sahinalp.
\newblock Not all scale-free networks are born equal: the role of the seed graph in ppi network evolution.
\newblock {\em PLoS computational biology}, 3(7):e118, 2007.

\bibitem[KKR04]{KKR04}
Tali Kaufman, Michael Krivelevich, and Dana Ron.
\newblock Tight bounds for testing bipartiteness in general graphs.
\newblock {\em SIAM Journal on Computing}, 33(6):1441--1483, 2004.

\bibitem[MSOI{\etalchar{+}}02]{milo2002network}
Ron Milo, Shai Shen-Orr, Shalev Itzkovitz, Nadav Kashtan, Dmitri Chklovskii, and Uri Alon.
\newblock Network motifs: simple building blocks of complex networks.
\newblock {\em Science}, 298(5594):824--827, 2002.

\bibitem[MTW04]{marsaglia2004fast}
George Marsaglia, Wai~Wan Tsang, and Jingbo Wang.
\newblock Fast generation of discrete random variables.
\newblock {\em Journal of Statistical Software}, 11(3):1--11, 2004.

\bibitem[PCJ04]{prvzulj2004modeling}
Natasa Pr{\v{z}}ulj, Derek~G Corneil, and Igor Jurisica.
\newblock Modeling interactome: scale-free or geometric?
\newblock {\em Bioinformatics}, 20(18):3508--3515, 2004.

\bibitem[PDFV05]{palla2005uncovering}
Gergely Palla, Imre Der{\'e}nyi, Ill{\'e}s Farkas, and Tam{\'a}s Vicsek.
\newblock Uncovering the overlapping community structure of complex networks in nature and society.
\newblock {\em nature}, 435(7043):814--818, 2005.

\bibitem[PFL{\etalchar{+}}02]{pennock2002winners}
David~M Pennock, Gary~W Flake, Steve Lawrence, Eric~J Glover, and C~Lee Giles.
\newblock Winners don't take all: Characterizing the competition for links on the web.
\newblock {\em Proceedings of the national academy of sciences}, 99(8):5207--5211, 2002.

\bibitem[PPDSBLP16]{prat2016put}
Arnau Prat-P{\'e}rez, David Dominguez-Sal, Josep-M Brunat, and Josep-Lluis Larriba-Pey.
\newblock Put three and three together: Triangle-driven community detection.
\newblock {\em ACM Transactions on Knowledge Discovery from Data (TKDD)}, 10(3):1--42, 2016.

\bibitem[PR02]{PR}
Michal Parnas and Dana Ron.
\newblock Testing the diameter of graphs.
\newblock {\em Random Structures and Algorithms}, 20(2):165--183, 2002.

\bibitem[SCW{\etalchar{+}}10]{sala2010measurement}
Alessandra Sala, Lili Cao, Christo Wilson, Robert Zablit, Haitao Zheng, and Ben~Y Zhao.
\newblock Measurement-calibrated graph models for social network experiments.
\newblock In {\em Proceedings of the 19th international conference on World Wide Web (WWW)}, pages 861--870, 2010.

\bibitem[SM00]{swiegers2000new}
Gerhard~F Swiegers and Tshepo~J Malefetse.
\newblock New self-assembled structural motifs in coordination chemistry.
\newblock {\em Chemical reviews}, 100(9):3483--3538, 2000.

\bibitem[SV06]{sole2006network}
Ricard~V Sol{\'e} and Sergi Valverde.
\newblock Are network motifs the spandrels of cellular complexity?
\newblock {\em Trends in ecology \& evolution}, 21(8):419--422, 2006.

\bibitem[TT22]{TT_edges}
Jakub T{\v{e}}tek and Mikkel Thorup.
\newblock Edge sampling and graph parameter estimation via vertex neighborhood accesses.
\newblock In {\em Proceedings of the 54th Annual ACM Symposium on the Theory of Computing (STOC)}, pages 1116 -- 1129, 2022.

\bibitem[T\v22]{tetek_triangles}
Jakub T\v{e}tek.
\newblock {Approximate Triangle Counting via Sampling and Fast Matrix Multiplication}.
\newblock In {\em Proceedings of the 49th International Colloquium on Automata, Languages, and Programming, ({ICALP})}, pages 107:1--107:20, 2022.

\bibitem[Wal74]{walker1974new}
Alastair~J. Walker.
\newblock New fast method for generating discrete random numbers with arbitrary frequency distributions.
\newblock {\em Electronics Letters}, 10(8):127--128, 1974.

\bibitem[Wal77]{walker1977efficient}
Alastair~J. Walker.
\newblock An efficient method for generating discrete random variables with general distributions.
\newblock {\em ACM Transactions on Mathematical Software}, 3(3):253--256, 1977.

\end{thebibliography}

\ifnum\stoc=0

\newpage
\appendix

\section{Table of notations}\label{app:notation}

\begin{table}[ht]
\centering
\begin{tabular}{|c|p{13cm}|}
\hline
\textbf{Symbol} & \textbf{Meaning} \\ \hline
$G$ & A (simple) graph \\ \hline
$n$ & Number of vertices in $G$ \\ \hline
$m$ & Number of (ordered) edges in $G$  \\ \hline
$\om$ & An estimate of the number of edges $m$\\ \hline
$\davg$ & Average degree in $G$ \\ \hline
$d(v)$ & Degree of vertex $v$ in $G$ \\ \hline
$d_S(v)$ & Number of neighbors $v$ has in $S$ \\ \hline
$\F$ & A motif (small graph)\\ \hline
$k$ & Number of vertices in $\F$ \\ \hline
$\countF$ & Number of copies of $\F$ in $G$\\ \hline
$\onF$ & An estimate of $\countF$ \\ \hline
$\ogamma$ & Set to $\onF^{1/k}$ \\ \hline
$\f$ & A  copy of $\F$ (in $G$) \\ \hline
$\Cyk$ & A cycle over $k$ vertices \\ \hline
$\Cy$ & A copy of $\Cyk$ (in $G$) \\ \hline
$\ham(\f)$ & The set of Hamiltonian cycles in $\f$ (copies of $\Cyk$ in $\f$) -- see Definition~\ref{def:Ham} \\ \hline
$\hF$ & The size of $\ham(\f)$ (for every copy $\f$ of $\F$) -- see Definition~\ref{def:Ham} \\ \hline
$\Q$ & The subgraph induced by the set of vertices of the cycle $\Cy$\\ \hline
$\calD$ & A data structure  used by our algorithms -- see Definition~\ref{def:good-DS}\\ \hline
$\vec{x}$& A sequence of positive integers $\vec{x} = \langle x_1,\dots, x_r\rangle$ such that $\sum_{q=1}^r x_q = k$ for which there exists a sequence of corresponding paths that form a cover -- see Definition~\ref{def:fitting}\\ \hline
$\nu_{\vec{x}}(\Cy)$ & The number of different sequences 
    of paths     covering $\Cy$ that correspond to $\vec{x}$ -- see Definition~\ref{def:cover} \\ \hline
   $\nu(\Cy)$  & The sum over all $\vec{x}$ that fit $\Cy$, of $\nu_{\vec{x}}(\Cy)$ -- see Definition~\ref{def:fitting} \\ \hline
$\nu(\f)$ & $\nu(\f) = \sum_{\Cy \in \ham(\f)} \nu(\Cy)$ -- see Definition~\ref{def:number-fittings}
\\ \hline   
$\kappaF$ & An upper bound on the number of copies of $\F$ that are contained in a given subgraph and contain a given cycle -- see  Definition~\ref{def:nF}.  \\ \hline
\end{tabular}
\label{tab:notation}
\caption{Table of Notations}
\end{table}

\section{Missing details for Section~\ref{subsec:preprocess-and-samp-high}}\label{app:samp-high}

\begin{theorem}[A data structure for a discrete distribution (e.g., \cite{walker1974new, walker1977efficient, marsaglia2004fast}).]\label{thm:DS}

\sloppy
There exists an algorithm that receives as input a discrete probability distribution $P$ over $q$ elements 
and constructs in time 
$O(q)$ a data structure that can be used to sample from $P$ in constant time per sample.
\end{theorem}

\alg{alg:preprocess}{
	{\bf \preprocess$\;(n,\oeps, \delta, \ogamma,\om)$} 
	\begin{enumerate}
        \item Set $\od = \frac{\om}{n}$.
		\item Repeat $t=\lceil\log(\frac{2}{\delta})\rceil$ times: \label{step:preprocess-repeat}
		\begin{enumerate}
			\item Let $S_i$ be a multiset of $s=\sets$ vertices chosen uniformly, independently, 
			at random. 
			\item Query the degrees of all the vertices in $S_i$ and compute $m(S_i)=\sum_{u\in S_i}{d(v)}$.
		\end{enumerate}
		\item Let $S$ be the  multiset $S_i$ with minimum value of $m(S_{i})$. \label{alg:step-choose-S}
		\begin{enumerate}
			\item	If $m(S)>2s\cdot \od$ then \textbf{return} \emph{fail}. \label{step:ret-fail} 
			\item Else, set up a data structure\footnote{See Theorem~\ref{thm:DS}} $\calD$ that supports sampling each vertex $u\in S$ with probability $\frac{d(u)}{m(S)}\cdot \frac{m(S)}{2\cdot \od\cdot s}=\frac{d(u)}{2\cdot \od\cdot s}$. \label{alg:step-ds-for-S}
		\label{step:S}
		\end{enumerate}
		\item \textbf{Return} $(\calD)$.
	\end{enumerate}
}

\smallskip
We note that the following proofs are almost identical to those in~\cite{cliquesSicomp, EMR_multiple_edges} for their corresponding procedures, but we bring them here due to small changes in the statements and for the sake of completeness.

\begin{proofof}{Lemma~\ref{lem:preproc}}
First we prove that with probability at least $1-\delta/2$, for every sampled multiset $S_i$, the following holds. For every vertex $v$ such that $d(v)> \ogamma,$ we have that $d_S(v)\in (1\pm \oeps)\cdot s\cdot \frac{d(v)}{n}$, where for a set $S$,  $d_S(v)=|\Gamma(v)\cap S|$.
Fix an iteration $i\in[t]$, and let 
the $j\th$ sampled vertex in $S_i$ be denoted by $u_i^j$ (so that $u_i^j$ is a random variable).
For any fixed  vertex $v\in V$ satisfying 
$d(v) > \ogamma$ and for each index $j\in [s]$, let
\[
\chi_i^j(v)=
\begin{cases}
    1 &\text{ $u_i^j$ is a neighbor of $v$} \\
    0 &\text{ otherwise}
\end{cases}
\]
Observe that $\Pr\left[\chi_i^j(v)=1\right]=\frac{d(v)}{n}$ for every $j\in [s]$
and that $d_{S_i}(v)=|\Gamma(v)\cap S_i|=\sum_{j=1}^s\chi_i^j(v)$.  Thus,
$\Ex\left[d_{S_i}(v)\right]=s\cdot \frac{d(v)}{n}.$
Since the $\chi_i^j(u)$ variables are independent $\{0,1\}$ random variables,  
by the multiplicative Chernoff bound,
\begin{equation}\label{eq:dsv}
\Pr\left[\left|d_{S_i}(v)-\frac{s\cdot d(v)}{n}\right|\geq \oeps\cdot \frac{s\cdot d(v)}{n}\right]\leq 2\exp\left({-\frac{\oeps^2\cdot s\cdot d(v)}{3n}}\right)\leq\frac{\delta}{2nt},
\end{equation}
where the last inequality is by the setting of $s = \sets$, and since $d(v)> \ogamma$. 
By taking a union bound over all vertices with $d(v)> \ogamma$, it holds that with probability at least $1-\frac{\delta}{2t}$, $d_{S_i}(v)\in (1\pm \eps)\frac{s\cdot d(v)}{n}$ with probability at most $\frac{\delta}{2t}$.  By taking a union bound over all $t$ sampled multisets $\{S_i\}_{i\in[t]}$, with probability at least $1-\frac{\delta}{2}$ the above holds for all multisets $S_i$.

For every $i$, $\Ex[m(S_i)]=s\cdot \davg$. Therefore, by Markov's inequality, with probability at least $1/2$, $m(S_i)\leq 2\cdot s \cdot \davg$.
It follows that $\min_{i\in[t]}\{m(S_i)\}> 2\cdot s\cdot \davg$ with probability at most $(1/2)^{t}=\delta/2$.
Hence, with probability at least $1-\delta$, the selected multiset $S$ is $(\oeps,\ogamma,\om)$-degrees-typical.
By Theorem~\ref{thm:DS}, it is possible to construct in time $O(S)$ a data structure which allows sampling vertices $v\in S$ with probability $d(v)/m(S)$  at unit cost per sample.
Thus we get  that $\calD$ is an $(\oeps,\ogamma,\om)$-degrees-typical data structure.

The bound on the query complexity and running time follows directly from the description of the algorithm and Theorem~\ref{thm:DS}.
\end{proofof}

\alg{alg:sampHigh}{
	{\bf \sampHighVertex$\;(\calD,\ogamma)$} 
	\begin{enumerate}
		\item Use the data structure $\calD$ to sample a vertex $u \in S$. 
         \label{step:fromS}
		\item Sample a uniform neighbor $v$ of $u$.   \label{step:nbr-of-S}
		\item If $d(v) \leq \ogamma$ \textbf{return} \emph{fail}. Otherwise, return $v$.\label{step:fail-if-low}
	\end{enumerate}
}

\medskip We now prove Lemma~\ref{lem:sampHeavyVertex} regarding sampling of medium and high degree vertices.\ForFuture{restate the lemma}

\begin{proofof}{Lemma~\ref{lem:sampHeavyVertex}}
    Condition on $\calD$ being  $(\oeps,\ogamma, \om)$-degrees-typical for some $(\oeps,\ogamma, \om)$-degrees-typical multiset $S$, and let $v$ be a vertex with $d(v)> \ogamma$. Then 
    \[
    \Pr[v \text{ is returned}]=\sum_{u \in  
    \Gamma(v)\cap S}\frac{d(u)}{2\cdot|S|\cdot \od}\cdot \frac{1}{d(u)}=\frac{d_{S}(v)}{2\cdot |S|\cdot \od}.
    \]
    By the definition of an $(\oeps,\ogamma, \om)$-degrees-typical $\calD$, since $d(v)>\ogamma$ it holds that $d_S(v)\in (1\pm\oeps) \cdot |S|\cdot \frac{d(u)}{n}$. 
    Therefore, 
    $$\Pr[v \text{ is returned }]\in \frac{(1\pm \oeps)\cdot |S|\cdot \frac{d(v)}{n}}{2 \cdot |S|\cdot \frac{\om}{n}}=(1\pm \eps)\frac{d(v)}{2 \om}.$$
    By Step~\ref{step:fail-if-low},  vertices $v$ with $d(v)\leq \ogamma$ have zero probability of being returned.

    By Theorem~\ref{thm:DS}, sampling from the data structure takes $O(1)$ time. The rest of the steps also take constant queries and time: a degree and neighbor query in order to sample a uniform neighbor in the second step and an additional degree query in the last step.
\end{proofof}

\section{Extension to directed graphs}\label{app:ext}

We shortly describe how our algorithm can be adapted to also work for directed motifs.

Consider an adaptation of the standard query model to directed graphs where it is possible to query both incoming edges and outgoing edges, and pair queries are ordered 
(i.e., a query on the pair $(u,v)$ is answered positively if there is an edge directed from $u$ to $v$).

In this model we can apply our algorithms by simulating the attempted samplers (designed for undirected graphs) on the undirected version of the graph (this can be done by ignoring the direction on the edges and considering the degree of a vertex to be the sum of its in-degree and out-degree). If a attempted sampler returns a copy of the motif with the wrong orientation of the edges, then we consider this as a failure of the sampler. 
Thus, our results hold for any directed motif $\F$ that has a Hamiltonian cycle (which is \emph{not necessarily directed}).

We note that it 
might be more  natural to consider a model where one can only query for outgoing neighbors.
However, this model is strictly weaker, and sublinear results might not be attainable even when the motif count is high.
For example, consider a family of graphs where in every graph in the family there are two nodes, each with $n-2$ outgoing neighbors and an additional edge between them (and the rest of the $n-2$ vertices all have incoming degree two, and zero outgoing neighbors).
Then there are $n-2$ directed motifs of the form $(v_1\rightarrow v_2\rightarrow v_3, v_1\rightarrow v_3)$, but it requires $\Omega(n)$ samples to distinguish between graphs in this family, and a family of graphs that is almost identical, except that there is no edge connecting the two high degree vertices. 


\fi
\end{document}